\documentclass[letterpaper, 10 pt, journal]{ieeetran}
\usepackage{graphicx}
\usepackage{amsmath}
\usepackage{amssymb}

\usepackage{amsthm}
\usepackage{amsfonts}
\usepackage{dsfont} 
\usepackage{subfigure}
\usepackage{cite}
\usepackage{algorithm}
\usepackage{algpseudocode}
\usepackage{nicefrac}
\usepackage{eufrak}
\usepackage{upgreek}
\usepackage{mathrsfs}
\usepackage{enumitem}
\usepackage[scientific-notation=true]{siunitx}
\usepackage{array}

\usepackage{xcolor}
\definecolor{forestgreen}{rgb}{0.13, 0.55, 0.13}
\definecolor{orange}{rgb}{1,0.49,0}

\newtheorem{defn}{Definition}
\newtheorem{thm}{Theorem}
\newtheorem{lem}{Lemma}
\newtheorem{prop}{Proposition}
\newtheorem{cor}{Corollary}

\newtheorem{rem}{Remark}

\newcommand{\R}{\mathbb{R}}
\newcommand{\N}{\mathbb{N}}
\newcommand{\C}{\mathcal{C}}
\newcommand{\X}{\mathcal{X}}

\newcommand{\D}{\mathcal{D}}
\newcommand{\Sig}{\mathcal{S}}
\newcommand{\OO}{\Omega}
\newcommand{\Ob}{\mathcal{O}}
\newcommand{\sat}{\vDash}

\newcommand{\PP}{\mathcal{P}}
\newcommand{\RR}{\mathcal{R}}
\newcommand{\CC}{\mathbb{C}}

\begin{document}
\title{Verifying Contracts for Perturbed Control \\ Systems using Linear Programming}
\author{Miel Sharf,~Bart Besselink,~Karl Henrik Johansson
\thanks{M. Sharf and K. H. Johansson are with the Division of Decision and Control Systems, School of Electrical Engineering and Computer Science, KTH Royal Institute of
Technology, 10044 Stockholm, Sweden. They are also affiliated with Digital Futures (e-mail: {\tt\small \{sharf,kallej\}@kth.se}).\\
B. Besselink is with the Bernoulli Institute for Mathematics, Computer Science and
Artificial Intelligence, University of Groningen, 9700 AK Groningen, The Netherlands (e-mail: {\tt\small b.besselink@rug.nl}).\\
This work was supported in part by DENSO Automative Deutschland GmbH, in part by the Knut and Alice Wallenberg Foundation, in part by the Swedish Strategic Research Foundation, in part by the Swedish Research Council, and in part by the Wallenberg AI, Autonomous Systems and Software Program (WASP) funded by the Knut and Alice Wallenberg
Foundation.}
}

\maketitle
\begin{abstract}
Verifying specifications for large-scale control systems is of utmost importance, but can be hard in practice as most formal verification methods can not handle high-dimensional dynamics. Contract theory has been proposed as a modular alternative to formal verification in which specifications are defined by assumptions on the inputs to a component and guarantees on its outputs.
In this paper, we present linear-programming-based tools for verifying contracts for control systems. We first consider the problem of verifying contracts defined by time-invariant inequalities for unperturbed systems. We use $k$-induction to show that contract verification can be achieved by considering a collection of implications between inequalities, which are then recast as linear programs. We then move our attention to perturbed systems. We present a comparison-based framework, verifying that a perturbed system satisfies a contract by checking that the corresponding unperturbed system satisfies a robustified (and $\epsilon$-approximated) contract. In both cases, we present explicit algorithms for contract verification, proving their correctness and analyzing their complexity. We also demonstrate the verification process for two case studies, one considering a two-vehicle autonomous driving scenario, and one considering formation control of a multi-agent system.
\end{abstract}

\section{Introduction}\label{sec.Intro}
In recent years, modern engineering systems have become larger and more complex than ever, as large-scale systems and networked control systems have become much more common, and the ``system-of-systems" philosophy has become the dominant design methodology. Coincidentally, specifications regarding these systems have grown more intricate themselves, and asserting they are met is of utmost importance. Recently, several attempts have been made to adapt contract theory, which is a modular approach for software verification, to control dynamical systems. In this paper, we present a framework for assume/guarantee contracts for discrete-time dynamical control systems, and present computational tools for verifying these contracts for perturbed and unperturbed linear time-invariant (LTI) systems using linear programming (LP).

\subsection{Background and Related Work}
The problem of verification tasks one to find a proof that a certain model satisfies given specifications.
This problem is referred to as model checking in the fields of computer science and software engineering, where it has been studied extensively over the last few decades \cite{Wallace1989,Baier2008}.  There, the software package under test is usually converted to or abstracted by a finite transition system, on which specifications are usually put in the form of linear temporal logic formulae. It is well-known that any linear temporal logic formula can be transformed to an equivalent automaton \cite{Vardi1996}, meaning that standard procedures from automata theory can be used to verify that the given finite transition system satisfies the specifications. Namely, one checks whether the set of accepted languages by the negation of the formula and the trace of the finite transition system have a non-empty intersection \cite{Baier2008}. In practice, this check is done by finding a path with certain desired properties in the graph describing the product automaton, implying it is tractable even for systems with thousands or millions of states.

Over the years, various attempts were made to apply the framework of model checking for finite transition systems to verify specifications for control systems with continuous (and infinite) state space. The main tools used in all of them are \emph{abstraction} and \emph{simulation}, which are notions connecting control systems with continuous state-spaces and finite transition systems \cite{Tabuada2009,Belta2017}. Namely, verification for a continuous control system is achieved by (1) abstracting it by a finite transition system, and (2) applying the model checking framework to this finite abstraction. The correctness of the verification process stems from the fact that the finite transition system approximately (bi-)simulates the continuous control system \cite{Girard2007,Girard2009,Wongpiromsarn2010}. Unfortunately, the abstraction of continuous control systems relies on discretization of the state-space. Thus, these methods cannot handle systems with high-dimensional dynamics, due to the curse of dimensionality, as this collection of methods treats the system-under-test as a single monolithic entity. In particular, even minute changes to the system (e.g., replacing one the actuators with a comparable alternative) would require executing a completely new verification process.

As noted in the literature, scalable development of large-scale systems with intricate specifications requires a modular approach, i.e., a design methodology allowing different components or subsystems to be developed independently of one another \cite{Baldwin2006,Huang1998}. This philosophy can be achieved by verifying or designing each component on its own, while treating all other components as part of the (unknown) environment. In software engineering, design and verification are often modular by design; requirements for the software package are almost always defined in terms of modules, or even individual functions and methods. Moreover, each function or module can be verified on its own, independently of the other parts of the software \cite{Grumberg1994}. 
Perhaps the best example of the modular design philosophy in software engineering is contract theory \cite{Meyer1992,Benveniste2018}. Contract theory is a modular approach for software engineering, which explicitly defines assumptions on the input and guarantees on the output of each software component. It can be used to design and verify software components, and even automatically fix bugs in the code \cite{Pei2014}. 

On the contrary, this situation is significantly different for control systems. Control design is often non-modular, as it requires the designer to know an exact (or approximate) model for each component in the system. For example, even the most scalable distributed and decentralized control methods, such as \cite{Siljak2005,Rantzer2015}, require a single authority with complete knowledge of the system model in order to design the decentralized or distributed controllers, i.e., they do not follow this modular design philosophy. Recently, several attempts have been made to derive modular design procedures for control systems. Some methods try to "modularize" the previous procedure, which treated the system as a single monolithic entity, by considering composition-compatible notions of abstraction and simulation \cite{Meyer2017,Hussien2017, Saoud2018b, Zamani2018}. Another approach, which is geared toward safety specifications, is to search for a composition-compatible method to calculate invariant sets \cite{Smith2016, Nilsson2016,Chen2018}.

In recent years, several attempts have been made to adapt contract theory to a modular design and verification framework for control systems. It has been successfully applied to the design of the ``cyber" aspects of cyber-physical systems, see \cite{Nuzzo2014,Nuzzo2015} and references therein. More recently, several frameworks have been proposed for contract theory for dynamical control systems, see e.g., \cite{Besselink2019,Shali2021,Saoud2018,Saoud2019,Eqtami2019,Ghasemi2020,Saoud2021}. 
In \cite{Besselink2019,Shali2021}, the authors propose methods for prescribing contracts on continuous-time systems, and verify these contracts either using geometric control theory methods, or using behavioural systems theory, respectively. Discrete-time systems are considered in \cite{Saoud2018,Saoud2019,Eqtami2019, Ghasemi2020}, where assumptions are put on the input signal to the system, and guarantees are put on the state and the output of the system. However, prescribing guarantees on the state of the system goes against the spirit of contract theory, as the state of the system is an internal variable. Thus, we aim at presenting a contract-based framework for discrete-time dynamical control systems which does not refer to the state of the system, and present efficient computational tools for their verification.

\subsection{Contributions}
In this paper, we propose a novel framework for assume/guarantee contracts on discrete-time dynamical control systems. These contracts prescribe assumptions on the input to a system and guarantees on its output, relative to its input. We prescribe LP-based computational tools for verification of contracts defined by time-invariant linear inequalities, both for unperturbed and perturbed LTI systems. These computational tools are explicitly stated by Algorithm \ref{alg.VerifyCertainIota} (for unpertubed LTI systems) and Algorithm \ref{alg.VerifyUncertain} (for perturbed LTI systems). First, we present LP-based computational tools applicable to unperturbed LTI systems for a class of contracts defined by time-invariant linear inequalities. Second, and more importantly, we extend the verification framework also for perturbed LTI systems. 
To the knowledge of the authors, no works presenting a contract theory framework for perturbed systems currently exist. We also note that standard formal theory methods usually require special treatment when applied to perturbed or uncertain systems \cite{Sadigh2016,Shen2019,ApazaPerez2021}.

We first tackle the verification problem for unperturbed LTI systems. We use strong induction to show that the system satisfies the contract if and only if an infinite number of implications between inequalities hold (Theorem \ref{thm.Inductive}). These implications are then recast as linear programs, and we use $k$-induction \cite{Donaldson2011} to achieve verification by solving finitely-many linear programs, culminating in Algorithm \ref{alg.VerifyCertainIota}, for which we prove correctness and analyze its complexity (Theorems \ref{thm.CertainAlgCorrectness} and \ref{thm.ThetaInduciveLTI}). We then consider the problem of verifying that a perturbed system $\Sigma$ satisfies a contract $\C$ defined by time-invariant linear inequalities. We first show that $\Sigma$ satisfies the contract $\C$ if and only if the nominal counterpart $\hat{\Sigma}$ of $\Sigma$ satisfies a robustified contract $\C^\prime$ (Theorem \ref{thm.UncertainEquiv}). Ideally, we could then achieve a comparison-based procedure for verification, verifying that $\Sigma$ satisfies $\C$ by showing that the unperturbed system $\hat{\Sigma}$ satisfies $\C^\prime$, as we already have tools for the latter task.
Unfortunately, the contract $\C^\prime$ is defined by \emph{time-varying} linear inequalities, as the robustification of the guarantees at time $k$ corresponds to the worst-case behaviour of the perturbation up to time $k$. To alleviate this problem, we consider the most lenient time-invariant contract $\hat{\C}$ refining $\C^\prime$. Unfortunately, as $\hat{\C}$ depends on the perturbation for the entire time horizon, infinitely many robustification terms are necessary, rendering this approach intractable. To address this, we \emph{approximate} $\hat{\C}$ by a tractable under-approximation $\hat{\C}_\epsilon$ of arbitrary precision $\epsilon > 0$. As a result, we can verify that $\Sigma$ satisfies $\C$ by verifying that the unperturbed LTI system $\hat{\Sigma}$ satisfies the contract $\hat{\C}_\epsilon$, which is defined by time-invariant linear inequalities (Proposition \ref{prop.Taus}). Thus, verification can be achieved using the LP-based tools presented earlier in the paper, resulting in Algorithm \ref{alg.VerifyUncertain}. The computational complexity of the algorithm scales as $\log(1/\epsilon)$, meaning that even extremely small values of $\epsilon$ are tractable. We also study the assumptions and $\epsilon$-optimality of the suggested verification algorithm, see Section \ref{subsec.AnalysisUncertain}. These tools present a significant extension of our preliminary results, presented in the conference paper \cite{SharfADHS2020}, which only considered a significantly restricted class of contracts, and only unperturbed LTI systems.

The rest of the paper is structured as follows. Section \ref{sec.Background} presents the basics of the assume/guarantee framework, and also gives some background on polyhedral sets. Section \ref{sec.Cert} presents more general LP-based tools for verification for unperturbed LTI systems. Section \ref{sec.Uncertain} presents LP-based tools for verification for perturbed LTI systems. Finally, Section \ref{sec.CaseStudy} exemplifies the achieved tools for verification through case studies.

\paragraph*{Notation}
We denote the collection of natural numbers by $\N = \{0,1,2,\ldots\}$. For two sets $X,Y$, we denote their Cartesian product by $X\times Y$.
For a positive integer $n$, we denote the collection of all signals $\N \to \R^n$ by $\Sig^n$. For vectors $v,u \in \mathbb{R}^n$, we understand $v \le u$ as an entry-wise inequality. Moreover, we denote the Euclidean norm of a vector $v\in \R^n$ as $\|v\|$, and the operator norm of a matrix $P$ as $\|P\| = \sup_{v \neq 0} \frac{\|Pv\|}{\|v\|}$. The all-one vector is denoted by $\mathds{1}$, and the Minkowski sum of two sets $X,Y \subseteq \R^d$ is defined as $X+Y = \{x+y: x\in X,y\in Y\}$

Given a state-space system $(A,B,C,D)$, the observability matrix of depth $m$ is denoted by $\mathcal{O}_m = [C^\top,(CA)^\top,\ldots,(CA^m)^\top]^\top$. Moreover, the observability index $\nu$ is the minimal integer such that ${\rm rank}~\mathcal{O}_\nu = {\rm rank}~\mathcal{O}_{\nu+1}$. Moreover, given a state $x$ for the system, we let $p_\Ob(x)$ be the projection of $x$ on the observable subspace of the system. 

\section{Background} \label{sec.Background}
In this section, we present some basic notions about assume/guarantee contracts, as well as some basic facts about polyhedral sets.

\subsection{Assume/Guarantee Contracts}
We present several basic notions in the theory of abstract assume/guarantee contracts for dynamical closed-loop control systems. These have been previously presented in the preliminary work \cite{SharfADHS2020}, and are derived from \cite{Nuzzo2015,Benveniste2018}. Computational tools for these contracts will be given in the upcoming sections.

\begin{defn} \label{def.Systems}
A system $\Sigma$ {has an input $d\in \Sig^{n_d}$, output $y \in \Sig^{n_y}$, and state $x\in \Sig^{n_x}$. It is defined by a set $\X_0 \subseteq \R^{n_x}$ of initial conditions, matrices $A,B,C,D,E,F$ of appropriate dimensions, and two bounded sets $\PP \subseteq \R^{n_p}$, $\RR \subseteq \R^{n_r}$.} The evolution and observation are given by the following equations, which hold for any $k\in \N$:
\begin{align} \label{eq.GoverningEquations}
    x(0)&\in \X_0,\\\nonumber
    x(k+1) &= Ax(k) + Bd(k) + E\omega(k),~ \omega(k)\in \PP\\ \nonumber
    y(k) &= Cx(k) + Dd(k) + F\zeta(k),~\zeta(k)\in \RR
\end{align}
For signals $d\in \Sig^{n_d}$ and $y\in \Sig^{n_y}$, we write $y\in \Sigma(d)$ if there exists a signal $x\in \Sig^{n_x}$ such that $d(\cdot),x(\cdot),y(\cdot)$ satisfy \eqref{eq.GoverningEquations}.
\end{defn}
We include the set of allowable initial states $\X_0$ in the definition of a system, as otherwise we cannot discuss several important specifications. For example, asking whether the output of the system always lies inside a safe set is meaningless if we make no assumptions on the initial state, e.g., it is meaningless if the initial state lies outside the safe set.

\begin{rem} \label{rem.InitDepend}
Definition \ref{def.Systems} can be extended by allowing $\X_0$ to be dependent of $d(0)$. This is reasonable if the system tries to avoid an obstacle whose position is defined by $d(\cdot)$, assuming that the system does not start on top of the obstacle. This is also reasonable for systems trying to track $d(\cdot)$, assuming their initial tracking error is not too large. The methods presented in this paper work under this more general assumption. However, we consider the restricted definition to enhance readability. 
\end{rem}

\begin{rem}
Definition \ref{def.Systems} can also include non-linear systems, as the sets $\PP,\RR$ can also included unmodeled non-linear terms.
\end{rem}

We consider specifications on dynamical control systems in the form of assume/guarantee contracts, which prescribe assumptions on the input signal $d(\cdot) \in \Sig^{n_d}$ and issue guarantees on the output signal $y(\cdot) \in \Sig^{n_y}$, relative to the input signal:
\begin{defn} \label{defn.AG}
An assume/guarantee contract is a pair $(\D,\OO)$ where $\D \subseteq \Sig^{n_d}$ are the assumptions and $\OO \subseteq \Sig^{n_d} \times \Sig^{n_y}$ are the guarantees.
\end{defn}
In other words, we put assumptions on the input $d(\cdot)$ and demand guarantees on the input-output pair $(d(\cdot),y(\cdot))$. 

Assume/guarantee contracts prescribe specifications on systems through the notion of satisfaction:
\begin{defn}
We say that a system $\Sigma$ satisfies $\C = (\D,\OO)$ (or implements $\C$), and write $\Sigma \sat \C$, if for any $d\in \D$ and any $y\in \Sigma(d)$, we have $(d,y)\in \OO$.
\end{defn}

Another notion that will be of use to us is the notion of refinement. It considers two contracts defined on the same system, and compares them to one another:
\begin{defn} \label{def.refine}
Let $\C_i = (\D_i,\OO_i)$ be contracts for $i=1,2$, with the same input $d(\cdot) \in \Sig^{n_d}$ and the same output $y(\cdot)\in \Sig^{n_y}$. We say $\C_1$ \emph{refines} $\C_2$ (and write $\C_1 \preccurlyeq \C_2$) if $\D_1 \supseteq \D_2$ and $\Omega_1 \cap (\D_2 \times \Sig^{n_y}) \subseteq \Omega_2 \cap(\D_2 \times \Sig^{n_y})$.
\end{defn}
Colloquially, $\C_1 \preccurlyeq \C_2$ if $\C_1$ assumes less than $\C_2$, but guarantees more given the assumptions. 

The framework of assume/guarantee contracts supports modularity in design using the notions of refinement and composition. These allow one to dissect contracts on composite systems to contracts on subsystems or on the individual components. The reader is referred to the references \cite{SharfADHS2020,Nuzzo2015} for more information about these notions. Moreover, the references \cite{SharfADHS2020} present preliminary results for computational tools verifying them. Due to space limitations, we focus on prescribing tools for verifying that a given system $\Sigma$ satisfies a given contract $\C$, without assuming that the system can be separated into smaller subsystems.  

\subsection{Polyhedral Sets} \label{subsec.Polyhedral}
In this paper, we focus on specifications defined by linear inequalities, i.e., specifications defined using polyhedral sets:
\begin{defn}
A set $S \subseteq \R^{d}$ is called polyhedral if it is defined by the intersection of finitely many half-spaces. Equivalently, there exist a matrix $A$ and a vector $b$ such that the set $S$ is defined by $S = \{z\in \R^d: Az \le b\}$.
\end{defn}
Polyhedral sets are known to be convex. Moreover, optimizing linear cost functions {over} them corresponds to solving a linear program, which can be done quickly using off-the-shelf solvers, e.g., Yalmip \cite{Lofberg2004}. Any polyhedral set has an equivalent representation, known as the vertex representation:
\begin{lem}[\hspace{0.1pt}\cite{Schrijver1998}]
The set $S \subseteq \R^d$ is polyhedral if and only if there exist matrices $F,G$ such that $S = \{F\lambda + G\theta : \mathds{1}^\top \lambda = 1,~\lambda,\theta \ge 0\}$.
\end{lem}
Both representations of the polyhedral set can be useful for different reasons. The subspace representation $\{Az \le b\}$ is usually easier to define, and can be used to easily calculate the pre-image of a polyhedral set under a linear transformation. The vertex representation is useful for computing the Minkowski sum of two polyhedral sets, and for computing the image of a polyhedral set under a linear transformation.

In this paper, we will often encounter a situation in which we would like to verify that one polyhedral set is a subset of another polyhedral set. This inclusion can be easily verified:
\begin{lem} \label{lem.Inclusion}
Let $S_1,S_2$ be polyhedral sets.
\begin{itemize}
    \item If the sets are given in subspace representation, $S_i = \{z \in \R^d: A_i z \le b_i\}$, then $S_1 \subseteq S_2$ if and only if $\varrho_j \le 0$ for any $j$, where $\varrho_j$ is given as the value of the following linear program, and ${\rm e}_j$ is the $j$-th standard basis vector:
    \begin{align*}
        \varrho_j = \max\{{\rm e}^\top_j (A_2 z - b_2) : A_1 z \le b_1\}.
    \end{align*}
    \item If the sets are given in vertex representation, $S_i = \{F_i \lambda + G_i \theta : \mathds{1}^\top \lambda = 1,~\lambda,\theta \ge 0\}$, then $S_1 \subseteq S_2$ if and only if there exist matrices $\Lambda,\Theta_F,\Theta_G$ {with positive entries} such that the following relations holds:
    \begin{align} \label{eq.InclusionVertex}
        {G_1 = G_2\Theta_G,~F_1 = F_2\Lambda + G_2\Theta_F,~\Lambda^\top \mathds{1} = \mathds{1}.}
    \end{align}
\end{itemize}
\end{lem}
\begin{proof}
The first claim follows immediately, as $\varrho_j \le 0$ if and only if $\{z:A_1z \le b_1\} \subseteq\{z:{\rm e}_j^\top A_2 z \le {\rm e}_j^\top b_2\}$. As for the second claim, \eqref{eq.InclusionVertex} holds if and only if:
\begin{itemize}
    \item The columns of $F_1$ belong to $S_2$.
    \item The columns of $G_1$ belong to $\{G_2\theta : \theta \ge 0\}$.
\end{itemize}
It is clear that if these conditions hold, then $S_1 \subseteq S_2$. As for the other direction, the first condition obviously holds. For the second condition, if we take some column $g$ of $G_1$, then $F_1{\rm e}_1 + tg \in S_1$ for any $t>0$. Thus, there exist some $\lambda_t,\theta_t\ge 0$ such that $\lambda_t^\top \mathds{1} = 1$ and $F_1{\rm e}_1 + tg = F_2 \lambda_t + G_2 \theta_t$. Thus:
\begin{align*}
    \frac{1}{t} G_2 \theta_t = \frac{1}{t} F_1{\rm e_1} - \frac{1}{t} F_2 \lambda_t + g
\end{align*}
As $t \to \infty$, the right hand side tends to $g$ (as the elements of $\lambda_t$ are bounded between $0$ and $1$). This means that $g$ lies in the closure of the closed set $\{G_2\theta : \theta \ge 0\}$. As $g$ was an arbitrary column of $G_2$, the proof is complete.
\end{proof}

\subsection{Linear Time-Invariant Contracts}
We aim to present efficient LP-based methods for verifying that $\Sigma \sat \C$, and we do so for contracts defined by linear inequalities:
\begin{defn}\label{def.LinCon}
A \emph{linear time-invariant} (LTI) contract $\C = (\D,\OO)$ of depth $m\in \N$ with input $d(\cdot) \in \Sig^{n_d}$ and output $y(\cdot) \in \Sig^{n_y}$ is given by matrices $\mathfrak A^r\in \R^{n_a \times n_d},\mathfrak G^r \in \R^{n_g\times (n_d+n_y)}$ for $r=0,\ldots,m$ and vectors $\mathfrak a^0 \in \R^{n_a},\mathfrak g^0 \in \R^{n_g}$ such that:
\begin{align} \label{eq.LinCon}
    \D &= \left\{d(\cdot): \sum_{r=0}^m\mathfrak A^r d(k-m+r) \le \mathfrak a^0,~\forall k\ge m\right\},\\ \nonumber
    \OO &= \left\{(d(\cdot),y(\cdot)): \sum_{r=0}^m\mathfrak G^r \begin{bmatrix} d(k-m+r) \\ y(k-m+r) \end{bmatrix} \le \mathfrak g^0,~ \forall k\ge m\right\}
\end{align}
\end{defn}
\begin{rem}
Definition \ref{def.LinCon} generalizes the contracts considered in \cite{SharfADHS2020}, which considered LTI contracts of depth $m=1$. It is no restriction to assume $m\ge 1$, as any contract of depth $m=0$ is also a contract of depth $m=1$ with $\mathfrak A^1, \mathfrak G^1 = 0$.
\end{rem}

Linear time-invariant contracts are defined using polyhedral sets for the stacked input vector $[d(k)^\top,\ldots,d(k-m)^\top]^\top$ and the similarly defined stacked output vector. We assume that the inequalities defining the assumptions are self-consistent, in the sense that if a signal satisfies them for some interval of length $m$, it can be extended for all future time.
\begin{defn}
Given matrices $\{\mathfrak V^r\}_{r=0}^m$ and a vector $\mathfrak v^0$, we say $(\{\mathfrak V^r\}_{r=0}^m,\mathfrak v^0)$ is \emph{extendable} if for any vectors $u_0,u_1,\ldots,u_m$ such that $\sum_{r=0}^m \mathfrak V^ru_r\le \mathfrak v^0$, there exists some vector $u_{m+1}$ such that $\sum_{r=0}^m \mathfrak V^r u_{r+1} \le \mathfrak v^0$.
\end{defn}
\begin{prop}
Let $\{\mathfrak V^r\}_{r=0}^m$ be matrices and $\mathfrak v^0$ be a vector. Write $\mathfrak{V}_- = \left[\mathfrak V^0,\ldots,\mathfrak V^m \right]$ and consider the polyhedral set $S_- = \{z: \mathfrak V_- z \le \mathfrak v^0\}$. We define the shift operators as:
\begin{align*}
    T = \left[\begin{smallmatrix} 0 & I & \cdots & 0 & 0 \\ 0 & 0 & \ddots & 0 & 0 \\ \vdots & \vdots & \ddots & \ddots & \vdots \\ 0 & 0 & \cdots & 0 & I \\ 0 & 0 & \cdots & 0 & 0 \end{smallmatrix}\right], ~ K = \left[\begin{smallmatrix} 0 \\ 0 \\ \vdots \\ 0 \\ I \end{smallmatrix}\right],
\end{align*}
where $I$ is the identity matrix. The tuple $(\{\mathfrak V^r\}_{r=0}^m,\mathfrak v^0)$ is extendable if and only if the polyhedral set $TS_-=\{Tz: z\in S_-\}$ is contained in the polyhedral set $S_- + {\rm Im} K $.
\end{prop}
In particular, extendibility can be tested using the tools presented in the previous subsection.

\begin{proof}
By writing $z = [u_0^\top,\ldots,u_{m}^\top]^\top$, extendibility is equivalent to the following implication - whenever $z \in S_-$, there exists some $u_{m+1}$ such that $Tz + Ku_{m+1} \in S_-$. In other words, if $z\in S_-$, then there exists some $u_{m+1}$ such that $Tz \in S_- + \{-Ku_{m+1}\}$. This corresponds to $z\in S_-$ implying that $Tz \in S_- + {\rm Im}K$, proving the proposition.
\end{proof}

\section{Verification for Unperturbed Systems} \label{sec.Cert}
In this section, we consider the verification problem for unperturbed systems. These are closed-loop systems $\Sigma$ of the form \eqref{eq.GoverningEquations} for which the sets $\PP,\RR$ consist of a single element. Equivalently, these are affine dynamical control systems governed by the following equations
\begin{align} \label{eq.GoverningEquations2}
    x(0)&\in \X_0,\\\nonumber
    x(k+1) &= Ax(k) + Bd(k) + w,~\forall k \in \N\\
    y(k) &= Cx(k) + Dd(k) + v,~\forall k \in \N.\nonumber
\end{align}
where the vectors $w,v$ depend on the matrices $E,F$ and the sets $\PP,\RR$, each containing a single element. Throughout this section, we fix such a system, governed by \eqref{eq.GoverningEquations2}. Moreover, we fix matrices $\{\mathfrak A^r,\mathfrak G^r\}_{r=0}^m$ and vectors $\mathfrak a^0, \mathfrak g^0$ defining an LTI contract $\C$ of depth $m$ via \eqref{eq.LinCon}. Our goal is to find a computationally tractable method for verifying that $\Sigma \sat \C$. 

This section is split into two parts. First, we present an exact reachability-based procedure for verifying whether $\Sigma \sat \C$ holds. The procedure will require us to solve infinitely many linear programs, meaning it is intractable. The second part of this section will use induction (or more precisely, $k$-induction \cite{Donaldson2011}) to augment the verification procedure to be tractable, at the cost of making it conservative.

\subsection{Reachability-based Verification}
By definition, a control system $\Sigma$ satisfies a contract $\C = (\D,\OO)$ if for any admissible input $d(\cdot)\in \D$, and any trajectory $(d(\cdot),x(\cdot),y(\cdot))$ of $\Sigma$, we have $(d(\cdot),y(\cdot))\in \Omega$. If the guarantees $\Omega$ were independent of the input $d(\cdot)$, then this property can be understood in terms of reachability analysis - $\Sigma \sat \C$ if and only if the output of any trajectory of the system, with inputs taken from $\D$, lies in the set $\Omega$. For LTI contracts, the assumptions and guarantees are stated as a collection of requirements, one corresponding to each time $k$. The theorem below uses extendibility to convert this reachability-based criterion by a collection of implications that must be verified. This is a generalization of Theorem 3 in \cite{SharfADHS2020}.
\begin{thm}
\label{thm.Inductive}
Let $\C = (\D,\OO)$ be an LTI contract of depth $m\ge 1$ of the form \eqref{eq.LinCon}, and let $\Sigma$ be a system of the form \eqref{eq.GoverningEquations2}. Assume that $(\{\mathfrak A^r\}_{r=0}^m,\mathfrak a^0)$ is extendable. Then $\Sigma \sat \C$ if and only if for any $n\in \N, n\ge m-1$, the following implication holds: for any $d_0,d_1,\ldots,d_{n+1} \in \R^{n_d}$, any $x_0,x_1,\ldots,x_{n+1} \in \R^{n_x}$ and any $y_0,y_1,\ldots,y_{n+1}\in \R^{n_y}$, the condition:
\begin{align} \label{eq.FullFormDemand}
\begin{cases}
    x_0 \in \X_0,\\
     \sum_{r=0}^m \mathfrak G^r \left[\begin{smallmatrix} d_{k-m+r} \\ y_{k-m+r} \end{smallmatrix}\right] \le \mathfrak g^0,&\forall k=m,\ldots,n,\\
    \sum_{r=0}^m \mathfrak A^r d_{k-m+r} \le \mathfrak  a^0,&\forall k=m,\ldots,n+1, \\
    x_{k+1} = Ax_k + Bd_k + w,&\forall k=0,\ldots,n, \\
    y_k = Cx_k + Dd_k + v,&\forall k=0,\ldots,n+1,
\end{cases}
\end{align} 
implies:
\begin{align}\label{eq.FullFormResult}
\sum_{r=0}^m \mathfrak G^r \begin{bmatrix} d_{n+1-m+r} \\ y_{n+1-m+r} \end{bmatrix} \le \mathfrak g^0.
\end{align}
\end{thm}
In other words, satisfaction is equivalent to the following collection of statements, defined for all $n\in \N$ - if the initial conditions hold, the guarantees hold up to time $n$, and the assumptions and the dynamics hold up to time $n+1$, then the guarantees hold at time $n+1$. We now prove the theorem:

\begin{proof}
Suppose first that whenever \eqref{eq.FullFormDemand} holds, so does \eqref{eq.FullFormResult}, and take any $d\in \D$ and $y\in \Sigma(d)$. Our goal is to show that $(d,y)\in \OO$. As $d\in \D$, the following inequality holds for all $k\in \N, k\ge m$:
\begin{align*}
     \sum_{r=0}^m \mathfrak A^r d(k-m+r) \le \mathfrak a^0.
\end{align*}
Moreover, as $y\in \Sigma(d)$, there exists some signal $x(\cdot)$ so that \eqref{eq.GoverningEquations2} holds for all $k\in \N$. Thus, if we choose $d_k = d(k), x_k = x(k)$ and $y_k = y(k)$ for all $k=0,1,\ldots,n+1$ and use the implication $\eqref{eq.FullFormDemand}\implies\eqref{eq.FullFormResult}$, we conclude that \eqref{eq.FullFormResult} holds for any $n\ge m$ by induction on $n$. Thus, $(d,y)\in\OO$, and hence $\Sigma \sat \C$.

On the contrary, suppose that $\Sigma \sat \C$, and we wish to prove that \eqref{eq.FullFormDemand} implies \eqref{eq.FullFormResult}. We take $n\in \N, n \ge m$ and some $d_0,d_1,\ldots,d_{n+1}\in \R^{n_d}$, $x_0,x_1,\ldots,x_{n+1}\in \R^{n_x}$ and $y_0,y_1,\ldots,y_{n+1}\in \R^{n_y}$ such that \eqref{eq.FullFormDemand} holds, and show that \eqref{eq.FullFormResult} also holds. Suppose that we show that there exist signals $ d(\cdot)$ and $ y(\cdot)$ such that ${y} \in \Sigma({d})$, ${d} \in \D$ both hold, and $ d(k) = d_k,  y(k) = y_k$ also hold for all $k=0,1,\ldots,n+1$. In that case, we have that $( d, y)\in \OO$ as $\Sigma \sat \C$, which would imply the desired inequality at time $k=n+1$. Thus, it suffices to prove that such signals $ d, y$ exist.

Recall that $(\{\mathfrak A^r\}_{r=0}^m,\mathfrak a^0)$ was assumed to be extendable. As the inequality $\sum_{r=0}^m \mathfrak A^r d(k-m+r) \le \mathfrak a^0$ holds for all $k=m,\ldots,n+1$, we conclude that there exists a signal ${d} \in \D$ such that $ d(k) = d_k$ holds for $k=0,1,\ldots,n+1$. We define signals ${x},{y}$ as follows - for $k=0,1,\ldots,n+1$, we define ${x}(k) = x_k$ and ${y}(k) = y_k$. For $k\ge n+2$, we define $ x(k) = A x(k-1) + B d(k-1) + w$ and $ y(k) = C x(k) + D d(k) + v$. As we assumed that \eqref{eq.FullFormDemand} holds, we conclude that $ y\in \Sigma( d)$. We thus proved the existence of signals ${d},{y}$ satisfying ${y} \in \Sigma({d})$, ${d} \in \D$, and $ d(k) = d_k,  y(k) = y_k$ for $k=0,1,\ldots,n+1$. We deduce the implication holds, concluding the proof.
\end{proof}

Theorem \ref{thm.Inductive} allows one to prove that an unperturbed LTI system $\Sigma$ satisfies an LTI contract $\C$ by proving infinitely-many implications of the form $\eqref{eq.FullFormDemand}\implies\eqref{eq.FullFormResult}$. Moreover, these implications can be seen as one polyhedral set being a subset of another polyhedral set, and can thus be verified using the tools in Section \ref{subsec.Polyhedral}.
These prove that if the system satisfies the contract ``up to time $n$", then it satisfies it ``up to time $n+1$". However, using the theorem directly to verify satisfaction is infeasible, as there are infinitely many implications to prove. Section \ref{subsec.kInduction} below will show that it suffices to prove finitely many implications of the form $\eqref{eq.FullFormDemand}\implies\eqref{eq.FullFormResult}$ in order to verify satisfaction. Moreover, we can test the validity of these implications by recasting them as optimization problems, similarly to Lemma \ref{lem.Inclusion}. For any $n,p\in \N$ such that $n-p\ge m-1$, we consider the following optimization problem:
\begin{align} \label{eq.Prob_np_Lin}
    \max_{d_k,x_k,y_k} ~&~ \max_i \left[{\rm e}_i^\top\left( \sum_{r=0}^{m} \mathfrak G^{r} \begin{bmatrix}d_{n+1-m+r} \\ y_{n+1-m+r}\end{bmatrix} - \mathfrak g^0\right)\right]\\ \nonumber
    {\rm s.t.} ~&~  \sum_{r=0}^{m} \mathfrak G^r \begin{bmatrix}d_{k-m+r} \\ y_{k-m+r}\end{bmatrix} \le \mathfrak g^0\hspace{2pt}~~,\forall k=m+p,\ldots,n,\\ \nonumber
    ~&~\sum_{r=0}^{m} \mathfrak A^r d_{k-m+r} \le \mathfrak a^0~~~~~~,\forall k=m+p,\ldots,n+1, \\ \nonumber
    ~&~x_{k+1} = Ax_k + Bd_k + w\hspace{2pt}~,\forall k=p,\ldots,n, \\ \nonumber
    ~&~y_k = Cx_k + Dd_k + v~~~~~,\forall k=p,\ldots,n+1, \\ \nonumber
    ~&~ x_p \in \X_p, \nonumber
\end{align}
where $\X_p,~p=1,2,\ldots,n$  are sets to be defined later, and $\rm e_i$ are the standard basis vectors. We denote this problem by $V_{n,n-p}$ and its optimal value as $\theta_{n,n-p}$. Here, $p$ represents the first time we consider, $n$ represents the last time at which we know the guarantee holds, and $\ell = n-p$ is the length of history we consider. For $p=0$, the problem $V_{n,n}$ computes the ``worst-case violation" of the guarantee at time $n+1$, given that the assumptions and dynamics hold at times $0\ldots,n+1$ and that the guarantees hold at times $0\ldots,n$. Thus, Theorem~\ref{thm.Inductive} can be restated in the following form:
\begin{cor} \label{cor.Thetann}
Under the assumptions of Theorem~\ref{thm.Inductive}, $\Sigma \sat \C$ if and only if $\theta_{n,n} \le 0$ for all $n \in \N$, $n\ge m-1$.
\end{cor}

\begin{proof}
For any $n\in \N$, $\theta_{n,n}\le 0$ if and only if \eqref{eq.FullFormResult} holds whenever \eqref{eq.FullFormDemand} holds. The result now follows by applying Theorem \ref{thm.Inductive}.
\end{proof}

\subsection{Tractable Verification using $k$-induction} \label{subsec.kInduction}
Corollary \ref{cor.Thetann} proves it suffices to compute $\theta_{n,n}$ for all $n\in\N, n\ge m-1$ to in order to verify whether $\Sigma\sat \C$. As this requires solving infinitely many linear programs, the method is intractable. Moreover, we prefer to compute $\theta_{n,\ell}$ for small $\ell = n - p$, as this is a simpler problem with fewer variables.
The main difficulty in using $V_{n,\ell}$ for small $\ell$ is that it requires knowledge of the state trajectory $x(\cdot)$ at time $p = n - \ell$, captured in \eqref{eq.Prob_np_Lin} via the constraint $x_p \in \X_p$. This is simply reduced to the initial value $x_0 \in \X_0$ for the problems $V_{n,n}$. 

An efficient solution of $V_{n,\ell}$ for small $\ell$ requires a characterization of $\X_p$ satisfying the following criteria. First, we would like $\X_p$ to be \emph{independent} of $p$, as this will imply that verification can be done by solving a finite number of optimization problems (thus not requiring the computation of all $\theta_{n,n}$ as in Corollary \ref{cor.Thetann}). 
Second, $V_{n,\ell}$ is equivalent to $V_{n+1,\ell}$ where $\X_{p+1}$ is the image of $\X_p$ under the dynamics $x_{p+1} = Ax_p + Bd_p + w$. Thus, we search for $\X_p$ which is a robust invariant set, and specifically the  smallest robust invariant set containing $\X_0$. 
However, the smallest robust invariant set containing $\X_0$ might be fairly complex to explicitly state, implying that the optimization problem $V_{n,\ell}$ cannot be explicitly defined, let alone solved. For example, \cite{Fisher1988} studies the minimally robust invariant set containing $\X_0 = \{0\}$ for two-dimensional linear time-invariant systems, given in state-space form via $x^+ = Ax + Bd$. It shows that this minimally robust invariant set is polyhedral if and only if all of the eigenvalues of the system matrix $A$ are rational. We note that the problem $V_{n,n-p}$ is a linear program if and only if $\X_p$ is a polyhedral set, meaning that if we take $\X_p$ as the minimal robust invariant set, the rationality of the eigenvalues of $A$ determines whether $V_{n,p}$ is a linear program or not. This problem is escalated even further if we assume that the eigenvalues of $A$ are computed numerically, meaning we cannot determine their rationality. We can also try and find some robust invariant set containing $\X_0$, not necessarily the smallest one, but an explicit form is still hard to find. For example, \cite{Rakovic2005} tries to find a polyhedral robust invariant set containing $\X_0$, offering a very partial solution for $\X_0 = \{0\}$. A general solution to this problem is not known to the authors. 

We make a detour around the tractability problem for the robust invariant set by choosing $\X_p = \R^{n_x}$. This results in a more conservative test, in the sense that $\X_p$ is larger than necessary, and the demand $\theta_{n,\ell} \le 0$ becomes stricter. However, the resulting problems $V_{n,n-p}$ are linear programs. Moreover, this choice allows us to verify contract satisfaction by solving finitely many linear programs, as suggested by Algorithm \ref{alg.VerifyCertainNoIota} below. The algorithm chooses which problems $V_{n,\ell}$ to solve based on an input parameter $\iota$, which essentially acts as a truncation parameter. Indeed, it defines the maximal history depth for the problems $V_{n,\ell}$ solved by the algorithm, and also the highest number $n$ for which $\theta_{n,n}$ is computed. Theorem \ref{thm.CertainAlgCorrectness} below studies Algorithm \ref{alg.VerifyCertainNoIota}, proving its correctness and analyzing its complexity. Later, Theorem \ref{thm.ThetaInduciveLTI} will suggest a choice of the parameter $\iota$.

\begin{algorithm} [t]
\caption{Verification for Unperturbed Systems, Ver. 1}
\label{alg.VerifyCertainNoIota}
{\bf Input:} An LTI contract $\C = (\D,\OO)$ of the form \eqref{eq.LinCon}, of depth $m \ge 1$, an LTI system $\Sigma$ of the form \eqref{eq.GoverningEquations2}, and a number $\iota \in \N$ satisfying $\iota \ge m-1$.\\
{\bf Output:} A boolean variable $\mathfrak{b}_{\mathcal{C},\Sigma,\iota}$.
\begin{algorithmic}[1]
\State Consider $V_{n,\ell}$ defined in \eqref{eq.Prob_np_Lin}, with $\X_1 = \R^{n_x}$. Solve them for $(n,\ell) \in\{(k,k): m-1\le k\le \iota-1\} \cup \{(\iota+1,\iota)\}$, and let $\theta_{n,\ell}$ be their solution.
\If{All computed values $\theta_{n,\ell}$ are non-positive}
\State {\bf Return} $\mathfrak{b}_{\mathcal{C},\Sigma,\iota} = $ true.
\Else
\State {\bf Return} $\mathfrak{b}_{\mathcal{C},\Sigma,\iota} = $ false.
\EndIf
\end{algorithmic}
\end{algorithm}

\begin{thm} \label{thm.CertainAlgCorrectness}
Let $\C = (\D,\OO)$ be an LTI contract of depth $m\ge 1$ of the form \eqref{eq.LinCon}, let $\Sigma$ be a system of the form \eqref{eq.GoverningEquations2}, and let $\iota$ be a natural number satisfying $\iota \ge m-1$.
\begin{enumerate}
    \item If Algorithm \ref{alg.VerifyCertainNoIota} outputs $\mathfrak{b}_{\C,\Sigma,\iota} = $ true, then $\Sigma \sat \C$.
    \item The algorithm solves $\iota-m+3$ linear programs.
\end{enumerate}
\end{thm}

\begin{proof}
The second part is obvious, so we focus on the first. We use two claims to prove the first part of the theorem:
\begin{itemize}
    \item[i)] For any $n\in \N$, we have $\theta_{n,n}\le\theta_{n,n-1}\le\cdots\le \theta_{n,m-1}$.
    \item[ii)] For any $\ell \ge m-1$, we have $\theta_{\ell,\ell} \le \theta_{\ell+1,\ell} = \ell_{\ell+2,\ell} = \cdots$.
\end{itemize}
We first explain why these claims imply the first part of the theorem, and then prove that the claims hold. 

Assume both claims hold and that $\mathfrak{b}_{\C,\Sigma,i} =$ true, i.e., that $\theta_{n,n} \le 0$ for $n=m-1,\ldots,\iota$ and that $\theta_{\iota+1,\iota} \le 0$, and we prove that $\Sigma \sat \C$. By Corollary \ref{cor.Thetann}, it suffices to show that $\theta_{n,n} \le 0$ for $n\ge \iota + 1$. If $n \ge \iota+1$, then $\theta_{n,n} \le \theta_{n,\iota}$ (by the first claim), and $\theta_{n,\iota} = \theta_{\iota+1,\iota}$ (by the second claim). As $\theta_{\iota+1,\iota} \le 0$ by assumption, we conclude that $\theta_{n,n} \le 0$. As $n \ge \iota+1$ was arbitrary, we yield $\Sigma \sat \C$.

We now prove claim i). Fix some $p$ such that $1\le p \le n-m+1$, so that $\ell = n-p$ satisfies $n-1\ge \ell \ge m-1$. We can relate the problem $V_{n,n-p+1}$ to the problem $V_{n,n-p}$ by altering some of its constraints. We first remove the constraints that the guarantees, assumptions and dynamics hold at time $p-1$. We also note that while the problem $V_{n,n-p}$ restricts $x_p \in \R^{n_x}$, $V_{n,n-p+1}$ restricts $x_p$ to be achieved from the dynamics (via $x_{p-1} \in \R^{n_x}$ and $x_p = Ax_{p-1} + Bd_{p-1}+w$). Thus, $V_{n,n-p+1}$ has the same cost function as $V_{n,n-p}$, but stricter constraints. In particular, as both are maximization problems, we conclude that $\theta_{n,n-p+1} \le \theta_{n,n-p}$, as desired. 

As for claim ii), we similarly relate $V_{n+1,\ell}$ and $V_{n,\ell}$ by altering the names of $d_k,x_k,y_k$ to $d_{k+1},x_{k+1},y_{k+1}$, and changing the set of initial conditions from $\X_0$ to $\R^{n_x}$ (only if $n=\ell$)
Thus, $\theta_{\ell,\ell} \le \theta_{\ell+1,\ell} = \theta_{\ell+2,\ell} = \cdots$.
\end{proof}

Theorem \ref{thm.CertainAlgCorrectness} shows that Algorithm \ref{alg.VerifyCertainNoIota} can be used to verify that an unperturbed system $\Sigma$ satisfies an LTI contract $\C$. The algorithm uses a parameter $\iota \ge m-1$ dictating the number of linear programs solved by the algorithm. Namely, the first $\iota-m+1$ programs  deal with the initial conditions of the system, and the last program deals with the long-term behaviour of the system.
As $\iota$ becomes larger, the algorithm becomes less (over-)conservative - if $\iota_1 \le \iota_2$ then $\theta_{n,\iota_2} \le \theta_{n,\iota_1}$, so $\theta_{n,\iota_1} \le 0$ implies $\theta_{n,\iota_2} \le 0$. However, larger values of $\iota$ result in a larger number of linear programs, which are also more complex as they have more variables.
We must find a systematic way to choose the parameter $\iota$ effectively. As the following theorem shows, $\theta_{\iota+1,\iota}$ can be infinite for one value of $\iota$, but finite (and non-positive) for other values of $\iota$:

\begin{thm} \label{thm.ThetaInduciveLTI}
Let $\C = (\D,\OO)$ and $\Sigma$ be as in Theorem \ref{thm.CertainAlgCorrectness}. Let $\nu$ be the observability index of $\Sigma$, and define $\X_p = \R^{n_x}$ for all $p\neq 0$.
Define the sets: 
\begin{align*}
    \D_\star &= \left\{(d_0,d_1,\ldots,d_m): \sum_{r=0}^m \mathfrak A^r d_r \le \mathfrak a^0\right\},
    \end{align*}
    \begin{align*}
    \Omega_\star &= \left\{(d_0,y_0,\ldots,d_m,y_m) : \sum_{r=0}^m \mathfrak G^r \begin{bmatrix} d_r \\ y_r \end{bmatrix} \le \mathfrak g^0\right\}
\end{align*}
Assume $\D_\star$ is bounded, and that for any bounded set $E \subseteq \R^{(m+1)n_d}$, the intersection $\Omega_\star \cap (E\times \R^{(m+1)n_y})$ is bounded. Then $\theta_{n,\ell} < \infty$ for $n \ge \ell \ge \max\{m,\nu\}-1$, and $\theta_{n,\ell} = \infty$ if $n,\nu-1 > \ell \ge m-1$. 
\end{thm}

\begin{proof}
Define $\mu = \max\{m,\nu\}$. We first show that $\theta_{n,\mu-1} < \infty$, implying that $\theta_{n,\ell} < \infty$ for all $\ell \ge \mu - 1$ by claim i) in the proof of Theorem \ref{thm.CertainAlgCorrectness}. 
Consider a feasible solution
$\{d_k,x_k,y_k\}_{k=n-\mu+1}^{n+1}$ of $V_{n,\mu-1}$. Because $\D_\star$ is bounded, $(\D_\star \times \R^{(m+1)n_y}) \cap \Omega_\star$ is bounded. Thus, for some constant $M_0>0$, we have $\|d_k\|,\|y_k\|,\|d_{n+1}\| \le M_0$ for $k=n-\nu+1,\ldots,n$ (as $\mu \ge m$). 
However, as $\mu \ge \nu$, $p_\Ob(x_{n-\mu+1})$ can be achieved as a linear combination of $y_{n-\mu+1},\ldots,y_n$ and $d_{n-\mu+1},\ldots,d_n$ using $\Ob_\mu$. We thus find some $M_1>0$, depending on $M_0$ and $\Ob_\nu$, such that $\|p_\Ob(x_{n-\mu+1})\| \le M_1$. 
As $\|d_{n-\mu+k}\| \le M_0$ for all $k$, we yield $\|p_\Ob(x_{n-\mu+k})\| \le M_k$ for $k=1,2,\ldots,\mu+1$, where $M_k = \|A\| M_{k-1} + \|B\|M_0$. Thus $\|y_{n+1}\| \le \|C\| M_{n+1} + \|D\|M_0$, implying that the set of feasible solutions $\{d_k,x_k,y_k\}_k$ of $V_{n,\mu-1}$ is bounded, and therefore $\theta_{n,\mu-1} < \infty$. 

For the second part, we note that if $\nu-1 > \ell \ge m-1$, then $\nu > m$. In particular, claim ii) in the proof of Theorem \ref{thm.CertainAlgCorrectness} implies it suffices to show that $\theta_{n,\nu-2} = \infty$.
By definition of the observability index, ${\rm rank}~\Ob_{\nu} > {\rm rank}~\Ob_{\nu-1}$, implying there exists a non-zero vector $\xi \in \ker(\Ob_\nu)^\perp \cap \ker(\Ob_{\nu-1})$, so $CA^k \xi = 0$ for $k\le \nu-2$, but $CA^{\nu-1}\xi \neq 0$. Take any feasible solution $\{d_k,x_k,y_k\}_{k=n-\nu+2}^{n+1}$ and some $\alpha \in \R$ to be chosen later. Define a new solution
$\{\check d_k,\check x_k,\check y_k\}_k$ by $\check d_k = d_k,$
\begin{align*}
     \check x_k = \begin{cases} x_{n-\nu+2} + \alpha\xi & k=n-\nu+2, \\ A\check x_{k-1} + B\check d_{k-1} & {\rm else} \end{cases},
\end{align*}
and $\check y_k = C\check x_k + D\check d_k$. We have that $\check d_k = d_k$,$\check x_k = x_k$ and $\check y_k = y_k$ for any $k\le n$, thus $\{\check d_k,\check x_k,\check y_k\}_{k=n-\nu+2}^{n+1}$ forms a feasible solution of $V_{n,\nu-2}$.
Moreover, $\check y_{n+1} = y_{n+1} + \alpha CA^{\nu-1}\xi$. We claim that for any $M>0$, there exists some $\alpha$ such that the value of the cost function of $V_{n,\nu-2}$ for the feasible solution $\{\check d_k,\check x_k,\check y_k\}_{k=n-\nu+2}^{n+1}$ is at least $M$.

Consider the set $Q = \Omega_\star\cap(D_\star \times \R^{(m+1)n_y})$. By assumption, $Q$ is bounded, hence $(\check d_{n-m+2},\check y_{n-m+2},\ldots,\check d_{n+1},\check y_{n+1})\not\in Q$ for any $\alpha\in \R$ such that $|\alpha|$ is large enough. This is only possible if there exists some $i$ such that the $i$-th row of $\mathfrak G^m$, denoted $\mathfrak G_i^m$, satisfies $(\mathfrak G_i^m)^\top \left[\begin{smallmatrix} 0 \\ \xi \end{smallmatrix}\right] \neq 0$. If we denote the sign of $(\mathfrak G_i^m)^\top \left[\begin{smallmatrix} 0 \\ \xi \end{smallmatrix}\right] \neq 0$ as $\lambda$ and choose $\alpha = \lambda t$ for $t$ arbitrarily large, the value of the cost function grows arbitrarily large. Thus $\theta_{n,n-\nu+2} = \infty$.
\end{proof}

Theorem \ref{thm.ThetaInduciveLTI} suggests a value for the parameter $\iota$ when running Algorithm \ref{alg.VerifyCertainNoIota}. Indeed, it shows that for guarantees defined by compact sets, the algorithm always declares ``false" if $\iota < \max\{m,\nu\}-1$, no matter if $\Sigma \sat \C$ or $\Sigma \not\sat \C$. As we already stated before, larger values of $\iota$ result in a less (over-)conservative algorithm, but also in a more complex and slower algorithm. For that reason, we run Algorithm \ref{alg.VerifyCertainNoIota} with $\iota = \max\{m,\nu\}-1$. We explicitly state this in Algorithm \ref{alg.VerifyCertainIota}. 
\begin{algorithm} [h]
\caption{Verification for Unperturbed Systems, Ver. 2}
\label{alg.VerifyCertainIota}
{\bf Input:} An LTI contract $\C = (\D,\OO)$ of the form \eqref{eq.LinCon}, of depth $m \ge 1$, an LTI system $\Sigma$ of the form \eqref{eq.GoverningEquations2}.\\
{\bf Output:} A boolean variable $\mathfrak{b}_{\mathcal{C},\Sigma}$
\begin{algorithmic}[1]
\State Compute the observability index $\nu$ of $\Sigma$.
\State Run Algorithm \ref{alg.VerifyCertainNoIota} with $\iota = \max\{m,\nu\}-1$, outputting the answer $\mathfrak{b}_{\mathcal{C},\Sigma,\iota}$.
\State {\bf Return} $\mathfrak{b}_{\mathcal{C},\Sigma} = \mathfrak{b}_{\mathcal{C},\Sigma,\iota}$
\end{algorithmic}
\end{algorithm}

The correctness of Algorithm \ref{alg.VerifyCertainIota}, as well as an estimate on its complexity, follow from Theorem \ref{thm.CertainAlgCorrectness}.

\section{Verification for Perturbed Systems} \label{sec.Uncertain}
The previous section provides an efficient method for verifying that a given unperturbed LTI system satisfies a given contract. We now extend our results to dynamical control systems with perturbations in the form of process and measurement noise, prescribing LP-based methods for verifying satisfaction.

For this section, we fix an LTI contract $\C = (\D,\OO)$ of the form \eqref{eq.LinCon}. We also fix a system $\Sigma$ as in \eqref{eq.GoverningEquations}, where the sets $\PP,\RR$ correspond to process noise and measurement noise.
\begin{rem}
Suppose we want to verify that $\Pi \sat \C$ for some \emph{nonlinear} system $\Pi$. In many cases, it suffices to show that $\Sigma \sat \C$ for some LTI system $\Sigma$ with appropriately chosen process and measurement noise. For example, if $\Pi$ is governed by the equations $x(k+1) = x(k) + \sin(x(k))$ and $y(k) = x(k)$, we can consider the perturbed LTI system $\Sigma$ governed by the equations $x(k+1) = x(k) + \omega(k)$ and $y(k) = x(k)$, where $|\omega(k)|\le 1$. Trajectories of $\Pi$ are also trajectories of $\Sigma$, so verifying that $\Sigma \sat \C$ is sufficient to prove that $\Pi \sat \C$.
\end{rem}

We can consider an analogue of $V_{n,p}$ for the perturbed system $\Sigma$ and the contract $\C$, which would be of the form:
\begin{align} \label{eq.Vnp_noise}
    \max_{d_k,x_k,y_k} ~&~ \max_i \left[{\rm e}_i^\top\left( \sum_{r=0}^{m} \mathfrak G^{r} \begin{bmatrix}d_{n+1-m+r} \\ y_{n+1-m+r}\end{bmatrix} - \mathfrak g^0\right)\right]\\ \nonumber
    {\rm s.t.} ~&~  \sum_{r=0}^{m} \mathfrak G^r \begin{bmatrix}d_{k-m+r} \\ y_{k-m+r}\end{bmatrix} \le \mathfrak g^0~,\forall k=m+p,\ldots,n,\\ \nonumber
    ~&~\sum_{r=0}^{m} \mathfrak A^r d_{k-m+r} \le \mathfrak a^0~,\forall k=m+p,\ldots,n+1, \\ \nonumber
    ~&~x_{k+1} = Ax_k + Bd_k+E\omega_k~,\forall k=p,\ldots,n, \\ \nonumber
    ~&~y_k = Cx_k + Dd_k+F\zeta_k~,\forall k=p,\ldots,n+1, \\ \nonumber
    ~&~x_p \in \X_p,\\ \nonumber
    ~&~\omega_k\in\mathcal{P},\zeta_k\in\mathcal{R}~,\forall k=p,\ldots,n+1.
\end{align}
As before, the computational tractability of the problem depends on the set $\X_p$, as well as the sets $\mathcal{R}$ and $\mathcal{P}$. If $\X_p,\mathcal{P}$ and $\mathcal{R}$ are all defined by linear inequalities, we get a linear program. However, if the sets $\mathcal{P} ,\mathcal{R}$ are not defined by linear inequalities, we might get a nonlinear problem, or even a non-convex problem. For example, a uniform norm bound on the process noise, $\PP = \{\omega : \omega^\top P \omega \le \gamma^2\}$, yields a quadratic optimization problem with $n-p+1 = \ell+1$ quadratic constraints. Another case is when the perturbation stems from sensor noise, and the sensor provides an estimate on its magnitude. In that case, we can write $\omega = (\delta,\Delta)$ where $\delta \in \R^{n_\delta}$ is the sensor noise, $\Delta\in \R$ is the estimate on its size satisfying $0 \le \Delta \le 1$, and $\|\delta\| \le \Delta$ holds. In this case, the optimization problem turns out to be non-convex.

\subsection{Comparison-based Verification}
To avoid nonlinear (or non-convex) problems, we take a different approach. Intuitively, the perturbed system $\Sigma$ satisfies the contract $\C$ if and only if the nominal version of $\Sigma$, with no process or measurement noise, satisfies a robustified version of $\C$. The goal of this section is to make this claim precise. The nominal version of $\Sigma$, denoted $\hat{\Sigma}$, is governed by:
\begin{align} \label{eq.Noiseless}
    &x(0) \in \X_0,\\ \nonumber
    &x(k+1) = Ax(k) + Bd(k),~\forall k\in \N,\\ \nonumber
    &y(k) = Cx(k) + Dd(k),~\forall k\in \N.
\end{align}
The system $\hat{\Sigma}$ is an unperturbed LTI system, so checking whether it satisfies some LTI contract is possible using Algorithm \ref{alg.VerifyCertainIota}. 
The following theorem precisely defines the robustified version of $\C$:
\begin{thm} \label{thm.UncertainEquiv}
Let $\Sigma$ be a perturbed LTI system governed by \eqref{eq.GoverningEquations}, and let $\C$ be an LTI contract of the form \eqref{eq.LinCon}, where $\mathfrak G^i = [\mathfrak G^i_d,\mathfrak G^i_y]$ for $i=0,\ldots,m$. Define the auxiliary system $\hat \Sigma$ as \eqref{eq.Noiseless}, and let $T = \sum_{r=0}^m \mathfrak G_y^r C A^r$. 
The system $\Sigma$ satisfies $\C$ if and only if $\hat{\Sigma}$ satisfies the contract $\C^\prime = (\D,\OO^\prime)$, where:\small
\begin{align} \label{eq.Cprime}
    \OO^\prime = \{&(d(\cdot),y(\cdot))\in \Sig^{n_d}\times \Sig^{n_y}: \\&\sum_{r=0}^m \mathfrak G^r\begin{bmatrix}d(k-m+r) \\ y(k-m+r)\end{bmatrix} \le \mathfrak g^0 - \tau^k, ~\forall k\ge m\},\normalsize \nonumber
\end{align}\normalsize
and the $i$-th entry of the vector $\tau^k$ is given by $\tau^{\mathcal R}_i + \tau^{\mathcal P,{\rm e}}_i + \sum_{\varsigma=0}^{k-m-1} \tau^{\mathcal{P},{\rm m},\varsigma}_i$, where:
\vspace{-5pt}
\small
\begin{align} \label{eq.taus}
    &\tau^{\mathcal R}_i = \sum_{\ell=0}^m \max \left\{{\rm e}_i^\top \mathfrak G_y^\ell F \zeta_{\ell}:~\zeta_\ell \in \mathcal{R},~\forall \ell\right\},\\ \nonumber
    &\tau^{\mathcal P,{\rm e}}_i = \sum_{\ell = 0}^{m-1} \max\left\{\left({\rm e}_i^\top\sum_{r=\sigma+1}^m\mathfrak G_y^r CA^{r-1-\ell}E\right)\omega_\ell :~\omega_\ell \in \mathcal{P} \right\},\\ \nonumber
    &\tau^{\mathcal P,{\rm m},\varsigma}_i = \max \left\{{\rm e}_i^\top T A^\varsigma E \omega: ~\omega \in \mathcal{P}\right\},~\forall \varsigma\in\N.     
\end{align} \normalsize
\end{thm}

\begin{proof}
We fix some $d(\cdot)\in \D$ and consider a trajectory $(d(\cdot),x(\cdot),y(\cdot))$ of $\Sigma$. By definition, there exist some signals $\omega(\cdot),\zeta(\cdot)$ such that for any $k\in \N$, we have
\begin{align*}
\begin{cases}
    \omega(k)\in \PP,~\zeta(k) \in \RR,~x(0) \in \X_0,\\
    x(k+1) = Ax(k) + Bd(k) + E\omega(k),\\
    y(k) = Cx(k) + Dd(k) + F\zeta(k).
\end{cases}
\end{align*}
We consider the corresponding trajectory $(d(\cdot),\hat x(\cdot),\hat y(\cdot))$ of $\hat{\Sigma}$ with no process nor measurement noise, i.e., we define
\begin{align*}
\begin{cases}
    \hat x(0) = x(0),\\
    \hat x(k+1) = A\hat x(k) + Bd(k),~\forall k\in \N.\\
    \hat y(k) = C\hat x(k) + Dd(k),~\forall k\in \N.
\end{cases}
\end{align*}
It is clear that for any time $t\in \N$, we have $y(t) = \hat y(t) + \tilde y(t)$ where $\tilde y(t) = F\zeta(t) + \sum_{s=0}^{t-1}CA^{t-s-1}E\omega(s)$. Fixing a time $k\ge m$, the guarantee of the contract $\C$ can be written as
\begin{align*}
\sum_{r=0}^m \mathfrak G^r \begin{bmatrix}d(k-m+r) \\ \hat y(k-m+r)\end{bmatrix} + \sum_{r=0}^m \mathfrak G_y^r \tilde y(k-m+r) \le \mathfrak g^0,
\end{align*}
or equivalently, by plugging the exact form of $\tilde y$, as

\small
\begin{align} \label{eq.Ineq1}
    \sum_{r=0}^m &\mathfrak G^r \begin{bmatrix}d(k-m+r) \\  \hat y(k-m+r)\end{bmatrix} \le \mathfrak g^0 - \\& \nonumber \sum_{r=0}^m \mathfrak G_y^r \left(F\zeta(k-m+r) + \sum_{s=0}^{k-m+r-1}CA^{k-m+r-s-1}E\omega(s)\right)
\end{align} \normalsize
By replacing the order of summation, the double sum on the right-hand side of \eqref{eq.Ineq1} can be written as
\begin{align} \label{eq.UncertainTechnical}
    \sum_{r=0}^m &\mathfrak G_y^rF\zeta(k-m+r) +\\& \nonumber \sum_{s=0}^{k-1}\sum_{r=\max\{m+s+1-k,0\}}^m \mathfrak G_y^rCA^{k-m+r-s-1}E\omega(s).
\end{align}
We can break the second sum into two double sums, one from $s=0$ to $s=k-m-1$ (for which the sum on $r$ starts at $0$), and one from $s=k-m$ to $s=k-1$ (for which the sum on $r$ starts at $m+s+1-k$). We thus get
\begin{align*}
    \sum_{s=0}^{k-1}&\sum_{r=\max\{m+s+1-k,0\}}^m \left(\mathfrak G_y^rCA^{k-m+r-s-1}E\right)\omega(s) \\= &\sum_{s=0}^{k-m-1}\sum_{r=0}^m \mathfrak G_y^rCA^{k-m+r-s-1}E\omega(s) \\&\hspace{-9pt}+ \sum_{s=k-m}^{k-1}\sum_{r=m+s+1-k}^m \mathfrak G_y^rCA^{k-m+r-s-1}E\omega(s) 
\end{align*}
replacing the summation index, $\varsigma = k-m-1-s$ for the first double sum and $\sigma = s-k+m$ for the second double sum, we arrive at the following expression
\begin{align*}
    &\sum_{\varsigma=0}^{k-m-1}\left(\sum_{r=0}^m \mathfrak G_y^r CA^r\right) A^\varsigma E \omega(k-m-1-\varsigma) \\ +&\hspace{7pt}\sum_{\sigma = 0}^{m-1} \left(\sum_{r=\sigma+1}^m\mathfrak G_y^r CA^{r-1-\sigma}E\right)\omega(\sigma+k-m).
\end{align*}
We plug this expression in place of the double sum in \eqref{eq.UncertainTechnical}, which is then plugged into \eqref{eq.Ineq1}. Fixing $\hat y(\cdot)$, the inequality must hold for any choice of $\zeta(\cdot),\omega(\cdot)$ (corresponding to different trajectories of $\Sigma$ for the same input $d(\cdot)$). Optimizing over $\omega(t)\in \mathcal{P},\zeta(t)\in\mathcal{R}, \forall t$ gives
\begin{align*} 
    \sum_{r=0}^m \mathfrak G^r \begin{bmatrix}d(k-m+r) \\ \hat y(k-m+r)\end{bmatrix} \le \mathfrak g^0 - \tau^{\mathcal{R}} - \sum_{\varsigma=0}^{k-m-1} \tau^{\mathcal R,{\rm m},\varsigma} - \tau^{\mathcal{P},{\rm e}},
\end{align*}
concluding the proof.
\end{proof}
{
Loosely speaking, the result of Theorem \ref{thm.UncertainEquiv} transfers the effect of perturbations from the system to the contract. In (11), $\tau_i^\RR$ captures the effect of measurement noise, whereas $\tau_i^{\PP,\rm e}$ and $\tau_i^{\PP,\rm m,\sigma}$ account for the effect of process noise (whose effect propagates through time). The theorem therefore} prescribes a comparison-based method of asserting that a perturbed LTI system $\Sigma$ satisfies a given time-invariant contract $\C$. Namely, we can check that an auxiliary (unperturbed) LTI system $\hat \Sigma$ satisfies another, robustified contract $\C^\prime$. The contract $\C^\prime = (\D,\OO^\prime)$ is defined by \emph{time-dependent} linear inequalities, as the vector $\tau^k$ explicitly depends on $k$. As a result, the methods exhibited in Section \ref{sec.Cert} are ineffective, as they assume the contract is time-invariant. In the next section, we overcome this problem using refinement and approximation.

\subsection{Tractability through Refinement and Approximation}
In order to overcome the problem of time-dependence, and to use the methods of Section \ref{sec.Cert}, we refine $\C^\prime$ by a time-invariant contract $\hat{\C} = (\D,\hat \OO)$, where:
\begin{align} \label{eq.C_hat}
   \hat \OO = \{(&d(\cdot),y(\cdot))\in \Sig^{n_d}\times \Sig^{n_y}:\\\nonumber &\sum_{r=0}^m\mathfrak G^r\begin{bmatrix}d(k-m+r) \\ y(k-m+r)\end{bmatrix} \le \mathfrak g^0 - \tau^\infty, ~\forall k\ge m\},
\end{align}
and we define $\tau^\infty_i = \tau^{\mathcal R}_i + \tau^{\mathcal P,{\rm e}}_i + \sum_{\varsigma=0}^{\infty} \tau^{\mathcal{P},{\rm m},\varsigma}_i$. It is obvious that $\hat{\C} \preccurlyeq \C^\prime$ if $\tau^{\mathcal{P},{\rm m},\varsigma}_i \ge 0$, which is guaranteed if $0\in \PP$. In fact, $\C^\prime$ is the ``largest" or ``most lenient" time-invariant contract which refines $\C^\prime$.
However, this raises a new problem - computing the vector $\tau^\infty$ requires computing the $\tau^{\mathcal{P},{\rm m},\varsigma}_i$ for all $\varsigma\in \N$ and all $i$, i.e., it requires solving infinitely many optimization problems. We address this issue by truncating the infinite sum and overestimating its tail. This approach is formalized in the following theorem:
\begin{prop} \label{prop.Taus}
Suppose the assumptions of Theorem \ref{thm.UncertainEquiv} hold, and let $\hat{\C} = (\D,\hat{\OO})$ be as in \eqref{eq.C_hat}. Assume that for some $N_0 \in \N$, the matrix $A^{N_0}$ is contracting, i.e., the operator norm $\|A^{N_0}\|$ is strictly smaller than $1$. Moreover, assume that $\PP,\RR$ are bounded sets. Then for any $i$, $\tau^\infty_i < \infty$. Furthermore, define $M_\PP = \max_{\omega\in\mathcal{P}}\|\omega\|$ and $K_{A,N_0} = 1+\|A\|+\ldots+\|A^{N_0-1}\|$. Then for any $\epsilon > 0$, if we define a contract $\hat{\C}_\epsilon = (\D,\hat{\Omega}_\epsilon)$ by:
\begin{align} \label{eq.C_hat_eps}
   \hat \OO_\epsilon = \{(&d(\cdot),y(\cdot))\in \Sig^{n_d}\times \Sig^{n_y}:\\\nonumber &\sum_{r=0}^m\mathfrak G^r\begin{bmatrix}d(k-m+r) \\ y(k-m+r)\end{bmatrix} \le \mathfrak g^0 - \tau^{\epsilon}, ~\forall k\ge m\},
\end{align}
then $\hat{\C}_\epsilon \preccurlyeq \hat{\C}$, where the entries of the vector $\tau^\epsilon$ are defined by:
\begin{align*}
&\tau^{\epsilon}_i =  \tau^{\mathcal R}_i + \tau^{\mathcal P,{\rm e}}_i + \sum_{\varsigma=0}^{N(\epsilon,i)-1} \tau^{\mathcal{P},{\rm m},\varsigma}_i + \epsilon,
\end{align*}
and:

\vspace{-15pt}
\small
\begin{align} \label{eq.NeiStableA}
   &N(\epsilon,i) = \max\left\{\left\lceil N_0 \log_{\frac{1}{\|A^{N_0}\|}} \left( \frac{\|T^\top e_i\|\|E\|K_{A,N_0} M_\PP    }{(1-\|A^{N_0}\|)\epsilon}\right)\right\rceil,N_0\right\}
\end{align}
\end{prop}

\begin{proof}
It is enough to show that for any $i$, the inequality $\tau_i^\infty \le \tau_i^{\epsilon}$ holds, or equivalently, that  $\sum_{\varsigma = N(i,\epsilon)}^\infty \tau_i^{\mathcal{P},{\rm m},\varsigma} \le \epsilon$ holds.
For any $\varsigma \in \N$, we have:
\begin{align} \label{eq.TauBound}
\tau^{\mathcal P,{\rm m},\varsigma}_i &= \max \left\{e_i^\top T A^\varsigma E \omega: ~\omega \in \mathcal{P}\right\} \\&\le \nonumber \|A^\varsigma\|\|E\|\|T^\top e_i\| M_\PP
\end{align}
By using the inequality $\|A^{aN_0+b}\| \le \|A^{N_0}\|^a\|A\|^b$ for $a,b\in\N$, we conclude that:
\begin{align*}
    \sum_{\varsigma = N(i,\epsilon)}^{\infty} \|A^\varsigma\| &\le \left(\sum_{t=0}^{N_0-1} \|A\|^t\right) \left(\sum_{\varsigma= N(i,\epsilon)/N_0}^\infty \|A^{N_0}\|^\varsigma\right) \\&\le K_{A,N_0}\frac{\|A^{N_0}\|^{{N(i,\epsilon)}/N_0}}{1-\|A^{N_0}\|},
\end{align*}
where we use $\|A^{N_0}\| < 1$. Thus, we get:
\begin{align*}
    \sum_{\varsigma = N(i,\epsilon)}^\infty \tau_i^{\mathcal{P},{\rm m},\varsigma} \le \|T^\top e_i\|\|E\|K_{A,N_0}M_\PP \frac{\|A^{N_0}\|^{{N(i,\epsilon)}/N_0}}{1-\|A^{N_0}\|}
\end{align*}
plugging in $N(i,\epsilon)$ from \eqref{eq.NeiStableA}, we conclude that the expression on the right hand side is smaller or equal than $\epsilon$, as desired.
\end{proof}

Theorem \ref{thm.UncertainEquiv} and Proposition \ref{prop.Taus} suggest the following comparison-based algorithm for verification for perturbed systems, at least when the assumptions of Proposition \ref{prop.Taus} hold:
\begin{algorithm} [h]
\caption{Verification for Perturbed Systems}
\label{alg.VerifyUncertain}
{\bf Input:} An LTI contract $\C$ of the form \eqref{eq.LinCon}, a perturbed system $\Sigma$ of the form \eqref{eq.GoverningEquations}, and a conservatism parameter $\epsilon > 0$.\\
{\bf Output:} A boolean variable $\mathfrak{b}_{\mathcal{C},\Sigma}$.
\begin{algorithmic}[1]
\State Define the auxiliary noiseless system $\hat{\Sigma}$.
\For {each $i$},
\State Compute $N(\epsilon,i)$ as in \eqref{eq.NeiStableA}.
\State Compute $\tau^{\mathcal R}_i,\tau^{\mathcal P,{\rm e}}_i$ and $\tau^{\mathcal P,{\rm m},\varsigma}_i$ according to \eqref{eq.taus} for $\varsigma = 0,1,\ldots,N(\epsilon,i)-1$ . 
\State Compute $\tau^{\epsilon}_i =  \tau^{\mathcal R}_i + \tau^{\mathcal P,{\rm e}}_i + \sum_{\varsigma=0}^{N(\epsilon,i)-1} \tau^{\mathcal{P},{\rm m},\varsigma}_i + \epsilon$.
\EndFor
\State Define the contract $\hat{\C}_\epsilon = (\mathcal{D},\hat{\Omega}_\epsilon)$ as in \eqref{eq.C_hat_eps}.
\State Run Algorithm \ref{alg.VerifyCertainIota} for the system $\hat{\Sigma}$ and the contract $\hat{\C}_\epsilon$, outputting the answer $\mathfrak{b}_{\hat{\mathcal{C}}_\epsilon,\hat \Sigma}$ .
\State {\bf Return} $\mathfrak{b}_{\mathcal{C},\Sigma} = \mathfrak{b}_{\hat{\mathcal{C}}_\epsilon,\hat \Sigma}$ .
\end{algorithmic}
\end{algorithm}

\subsection{Properties of Algorithm \ref{alg.VerifyUncertain}} \label{subsec.AnalysisUncertain}
The rest of this section is devoted to studying the correctness, the assumptions, the conservatism, and the computational complexity of Algorithm \ref{alg.VerifyUncertain}. First, we claim that the algorithm correctly verifies satisfaction:

\begin{thm}[Correctness]
Suppose the assumptions of Theorem \ref{thm.UncertainEquiv} hold. If Algorithm \ref{alg.VerifyUncertain} outputs $\mathfrak{b}_{\C,\Sigma} =$ true, then $\Sigma \sat \C$.
\end{thm}
\begin{proof}
Algorithm \ref{alg.VerifyUncertain} outputs $\mathfrak{b}_{\C,\Sigma} =$ true if and only if Algorithm \ref{alg.VerifyCertainIota}, when applied on the nominal system $\hat{\Sigma}$ and the robustified contract $\hat{\C}_\epsilon$, outputs $\mathfrak{b}_{\hat\C_\epsilon,\hat\Sigma} =$ true. In that case, Theorem \ref{thm.CertainAlgCorrectness} implies that $\hat{\Sigma} \sat \hat{\C}_\epsilon$, hence $\hat{\Sigma} \sat \C^\prime$ as $\hat{\C}_\epsilon \preccurlyeq \hat{\C} \preccurlyeq \C^\prime$. Thus, Theorem \ref{thm.UncertainEquiv} implies that $\Sigma \sat \C$.
\end{proof}

We now study the assumptions of Algorithm \ref{alg.VerifyUncertain}, claiming they are not too strict.
\begin{thm}[Generality of Assumptions]\label{thm.AssumptionTau}
Suppose the assumptions of Theorem \ref{thm.UncertainEquiv} hold. Then:
\begin{itemize}
    \item There exists $N_0 \in \N$ such that $\|A^{N_0}\|<1$ if and only if $A$ is a strictly stable matrix, i.e., all of its eigenvalues are inside the open unit disc in the complex plane.
    \item Suppose that $A$ is not strictly stable, that $0\in \RR$, and that the set $\mathcal{P}$ contains a neighborhood of the origin. Suppose further that $E$ has full row rank and that the image of $T^\top$ is not contained within the stable subspace of $A$. Moreover, assume that for some $d_0,d_1,\ldots,d_m$, the following set is bounded and non-empty:
    \begin{align*}
        Q = \left\{(y_0,\ldots,y_m) : \sum_{r=0}^m \mathfrak G^r \begin{bmatrix} d_r \\ y_r \end{bmatrix} \le \mathfrak g^0\right\}.
    \end{align*}
    Then $\Sigma \not\sat \C$.    
\end{itemize}
\end{thm}

The first claim implies the algorithm is applicable for strictly stable systems, and the second shows that systems which are not strictly stable cannot satisfy compact specifications, at least generically (as the matrix $T$ depends on the constraints). 

\begin{proof}
We prove the claims in order. First, we denote the spectral radius of the matrix $A$ by $\rho(A) = \max\{|\lambda|: \exists v\neq 0, Av=\lambda v\}$. This is the maximum absolute value of an eigenvalue of $A$. By definition, $A$ is strictly stable if and only if $\rho(A)<1$. Moreover,  Gelfand's formula states that $\lim_{n\to \infty} \|A^n\|^{1/n} = \inf_{n\ge 1} \|A^n\|^{1/n}= \rho(A)$ \cite{Lax2002}. Thus, there exists some $N_0 \in \N$ such that $\|A^{N_0}\|<1$ if and only if $\rho(A) < 1$, as claimed.


The proof of the second claim is relegated to the appendix, as it is a bit more involved. We will, however, give a sketch of the proof here. First, we show that most entries of the vector $\tau^k$ grow arbitrarily large as $k$ grows to infinity. Namely, we show that the $i$-th entry grows arbitrarily large if $T^\top {\rm e}_i$ is outside the stable subspace of $A$. In the second stage, we use this to show that the inequality defining the set $\Omega^\prime$ in \eqref{eq.Cprime} defines an empty set if $k$ is large enough. We will then conclude from Theorem \ref{thm.UncertainEquiv} that $\Sigma \not \sat \C$.
\end{proof}

Next, we study the algorithm's approximation properties:
\begin{thm} \label{thm.epsTau}
Suppose that the assumptions of Theorem \ref{thm.UncertainEquiv} hold. Let $n,\ell\in \mathbb{N}$ such that $n\ge \ell \ge m-1$, and let $\epsilon > 0$. We denote the problems \eqref{eq.Prob_np_Lin} associated with $\hat \Sigma \sat \hat \C$ and $\hat \Sigma \sat \hat \C_\epsilon$ by and $V_{n,\ell}$ and $V_{n,\ell|\epsilon}$ respectively, and their values by $\theta_{n,\ell}$ and $\theta_{n,\ell|\epsilon}$. If $\theta_{n,\ell|\epsilon} > \epsilon$, then $\theta_{n,\ell} > 0$. In particular, Algorithm \ref{alg.VerifyCertainIota} would declare that $\hat\Sigma\not\sat \hat\C$.
\end{thm}
In other words, the parameter $\epsilon$ serves as a tunable conservatism parameter for the approximation $\hat{C}_\epsilon \preccurlyeq \hat{\C}$.
\begin{proof}
We let $\mathbb{X}_{n,\ell}$ denote the feasible set of $V_{n,\ell}$, and $\mathbb{X}_{n,\ell|\epsilon}$ denote the feasible set of $V_{n,\ell|\epsilon}$. By construction,  $\tau^{\epsilon}_i - \epsilon\le \tau^\infty_i \le \tau^{\epsilon}_i$ holds for every $i$. 
Thus, by definition of the contracts $\hat\C,\hat{\C}_\epsilon$, we conclude that $\mathbb{X}_{n,\ell} \supseteq \mathbb{X}_{n,\ell|\epsilon}$, as the constraints corresponding to the assumptions and the dynamics are identical, but the constraints corresponding to the guarantees are stricter.
Moreover, fixing some index $i$ in the cost function, the two problems ${V}_{n,\ell},{V}_{n,\ell|\epsilon}$ have the same cost function up to a constant, equal to $\tau^\infty_i - \tau^{\epsilon}_i$.

Choose an index $i$ such that at the optimal solution of ${V}_{n,\ell|\epsilon}$, the maximum of the cost function is attained at the index $i$. As both ${V}_{n,\ell|\epsilon}$ and $V_{n,\ell}$ are maximization problems, we yield:
\begin{align*}
    {\theta}_{n,\ell} \ge {\theta}_{n,\ell|\epsilon} + \tau^\infty_i - \tau^{\epsilon}_i \ge
    {\theta}_{n,\ell|\epsilon} - \epsilon > 0
\end{align*}
as claimed.
\end{proof}

Lastly, we shed light on the computational complexity of the algorithm. As before, we denote the depth of the contract $\C$ as $m$, and the observability index of the noiseless system $\hat{\Sigma}$ by $\nu$. The algorithm revolves around solving optimization problems of three different kinds:
\begin{itemize}
    \item[i)] Solving the linear programs determining whether $\hat{\Sigma} \sat \hat{\C}_\epsilon$. There are a total of $\max\{\nu,m\}+1$ linear programs, of dimension at most $(n_d+n_y+n_x)(\max\{\nu,m\}+1)$.
    \item[ii)] Solving $M_\PP = \max_{\omega\in \mathcal{P}}\|\omega\|$ to compute $N(\epsilon,i)$. 
    \item[iii)] Solving the optimization problems in \eqref{eq.taus}. We need to solve a total of $\sum_{i} (N(\epsilon,i)+2m+1)$ problems.
\end{itemize}
Solving the optimization problem (i) can be done very quickly using off-the-shelf optimization software, e.g., Yalmip \cite{Lofberg2004}. The tractability of the problems (ii) and (iii) depends on the exact form of $\mathcal{P},\mathcal{R}$. However, solving them is much more simple than solving \eqref{eq.Vnp_noise} for four main reasons:

First, these problems consider a single instance of $\omega$ or $\zeta$ at any given time, meaning that they are of a significantly lower dimension than \eqref{eq.Vnp_noise}, and they include far less constraints.

Second, the cost functions of these maximization problems are convex, meaning that the maximum is achieved on an extreme point of the set $\mathcal{P}$ or $\mathcal{R}$ \cite[Theorem 32.2]{Rockafellar1970}. Thus, even if the sets $\mathcal{P},\mathcal{R}$ are not convex, we can replace them by their convex hulls without changing the value of the problem. In other words, the convex relaxations of these optimization problems have the same value as the original problems.

Third, even if the sets $\mathcal{P},\mathcal{R}$ (or their convex hulls) are not defined using linear or quadratic inequalities, so standard LP and quadratic programming methods cannot be used, we can still use gradient-based, duality-based or interior-point-based methods. These methods will converge much faster for the optimization problems (ii) and (iii) than for the problem \eqref{eq.Vnp_noise}, due to the reduced dimension. 

Lastly, the simplicity of the optimization problems (ii) and (iii) allows one to give closed-form formulae for the solution if $\mathcal{P},\mathcal{R}$ are described using simple terms, thus eliminating the need for a numerical solution of the problems. Indeed, the following proposition gives closed-form solution to the optimization problems appearing in \eqref{eq.taus} and in Proposition \ref{prop.Taus}:
\begin{prop} \label{prop.TauExplicit}
Consider a set $\mathcal{H} \subseteq \R^{q}$. We take a vector $b \in \R^{q}$, and define $M_b = \max_{z\in \mathcal{H}} b^\top z$ and $M_{\|} = \max_{z\in \mathcal{H}} \|z\|$.
\begin{itemize}
    \item If $\mathcal{H} = \{z: z^\top H z \le \gamma^2\}$ for some positive-definite matrix $H$ and $\gamma > 0$, then $M_b = \gamma \|H^{-1/2}b\|$ and $M_\| = \gamma \|H^{-1/2}\|$.
    \item If $\mathcal{H}$ is a bounded polyhedral set given in vertex representation, $\mathcal{H} = \{F\lambda : \mathds{1}^\top \lambda = 1, \lambda \ge 0\}$, then $M_b = \max_i {\rm e}_i^\top F^\top b$ and $M_\| = \max_i \|F{\rm e}_i\|$
\end{itemize}
\end{prop}
\begin{proof}
For the first case, we note that $z^\top H z \le \gamma^2$ if and only if $\|v\| \le \gamma$, where $z = H^{-1/2}v$. Thus:
\begin{align*}
    M_b &=  \gamma\max_{\|v\|\le 1} (H^{-1/2}b)^\top v = \gamma \|H^{-1/2}b\|,\\
    M_\| &=  \gamma\max_{\|v\|\le 1} \|H^{-1/2}v\|= \gamma \|H^{-1/2}\|
\end{align*}
For the second case, the result follows from the fact that the maximum of a convex function on a bounded polyhedral set is attained at one of its vertices \cite[Theorem 32.2]{Rockafellar1970}.
\end{proof}

We make one last remark about the number $N(\epsilon,i)$, which dictates the number of problems \eqref{eq.taus} we have to solve.
\begin{rem}
In Algorithm \ref{alg.VerifyUncertain}, we compute $N(\epsilon,i)$ using \eqref{eq.NeiStableA}, which depends on a number $N_0$ such that $\|A^{N_0}\| < 1$. First, the number $N(\epsilon,i)$ depends logarithmically on $1/\epsilon$, meaning that the algorithm is computationally tractable even for extremely small values of $\epsilon$. Second, if $A$ is strictly stable, then there exist infinitely many $N_0$ such that $\|A^{N_0}\| < 1$. Moreover, $N(\epsilon,i) \ge N_0$ holds by definition. Thus, we can iterate over different values of $N_0$ to find the smallest possible value of $N(\epsilon,i)$ for fixed $\epsilon$ and $i$. See Algorithm \ref{alg.ComputeNepsi} for details. 
\end{rem}

\begin{algorithm} [h]
\caption{Computing the Optimal Threshold $N(\epsilon,i)$}
\label{alg.ComputeNepsi}
{\bf Input:} A stable matrix $A$, a matrix $C$, matrices $\{\mathfrak G_y^r\}_{r=0}^m$, a perturbation set $\mathcal{P}$, and a parameter $\epsilon > 0$\\
{\bf Output:} An optimal value of $N(\epsilon,i)$.
\begin{algorithmic}[1]
\State Compute $T = \sum_{r=0}^m \mathfrak G_y^r C A^r$ and $M_\PP = \max_{\omega \in \mathcal{P}} \|\omega\|$.
\State Put $N_0 = 1$, $N^{\epsilon,i}_{\rm opt} = \infty$, and $K_{A,N_0} = 0$.
\While{$N_0 \le N^{\epsilon,i}_{\rm opt}$}
\State Add $\|A^{N_0 - 1}\|$ to the value of $K_{A,N_0}$.
\If{$\|A^{N_0}\| < 1$}
\State Compute $N(\epsilon,i)$ according to \eqref{eq.NeiStableA}.
\State Assign the value $\min\{N^{\epsilon,i}_{\rm opt},N(\epsilon,i)\}$ to $N^{\epsilon,i}_{\rm opt}$.
\EndIf
\State Assign the value $N_0+1$ to $N_0$.
\EndWhile
\State {\bf return} $N_{\rm opt}^{\epsilon,i}$
\end{algorithmic}
\end{algorithm}

\section{Numerical Examples} \label{sec.CaseStudy}
In this section, we apply the presented verification algorithm in two case studies. The first deals with a two-vehicle autonomous driving scenario, and the second deals with formation control for multi-agent systems. 

\subsection{Two-Vehicle Leader-Follower system}
We consider two vehicles driving along a single-lane highway, as in Fig. \ref{fig.LeaderFollower}. We are given a headway $h>0$, and our goal is to verify that the follower keeps at least the given headway from the leader. Denoting the position and velocity of the follower as $p_f(k)$, $v_f(k)$, and the position and velocity of the leader as $p_l(k),v_l(k)$, the follower vehicle keeps the headway if and only if $p_f(k) - p_l(k) - hv_l(k) \ge 0$ holds at any time $k\in \N$. This scenario has been studied in \cite{SharfADHS2020} where the follower is assumed to have a known and unperturbed model. Here, we instead consider the same scenario for a follower with a perturbed model, affected by process noise.

\begin{figure}[b]
    \centering
    \includegraphics[width = 0.5\textwidth]{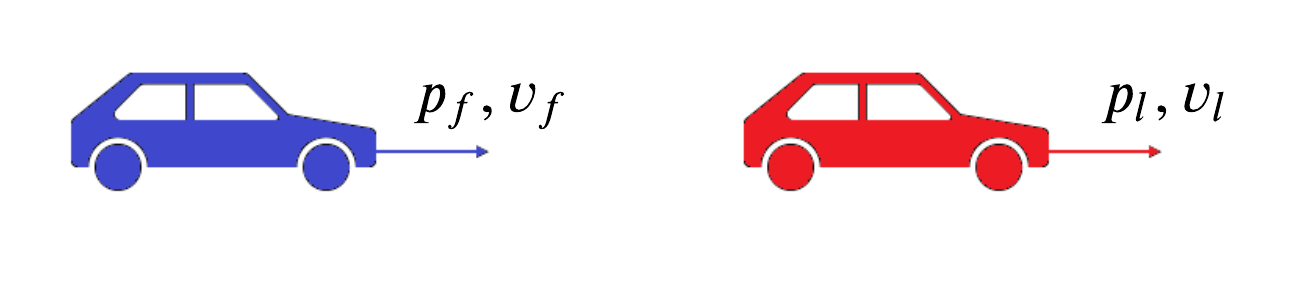}
    \caption{Two vehicles on a single-lane highway.}
    \label{fig.LeaderFollower}
\end{figure}

We start by explicitly stating the contract on the follower. The input to the follower includes the position and velocity of the leader, i.e., $d(k) = [p_l(k),v_l(k)]^\top$. The output from the follower includes its position and velocity, i.e., $y(k) = [p_f(k),v_f(k)]^\top$. 
For assumptions on the input, we assume the leader vehicle follows the kinematic laws with a bound on the acceleration, i.e., for any time $k$,
\begin{align*}
&p_l(k+1) = p_l(k) + \Delta t v_l(k),~v_l(k+1) = v_l(k) + \Delta t a_l(k),~\\
&a_l(k) \in [-a_{\rm min},a_{\rm max}],
\end{align*}
where $a_l(k)$ is the acceleration of of the leader vehicle at time $k$, and $\Delta t>0$ is the length of the discrete time-step. For guarantees, we specify that the headway is kept, i.e., that  $p_l(k) - p_f(k) - hv_f(k) \ge 0$ holds for any $k\in \N$. These specifications define a linear time-invariant contract $\C$ of depth $m=1$, defined using the following matrices and vectors:
\begin{align*}
&\mathfrak A^1 = \begin{bmatrix} 1 & 0 \\ -1 & 0 \\ 0 & 1 \\ 0 & -1 \end{bmatrix},~
\mathfrak A^0 = \begin{bmatrix} -1 & -\Delta t \\ 1 & \Delta t \\ 0 & -1 \\ 0 & 1  \end{bmatrix},~
\mathfrak a^0 = \begin{bmatrix} 0 \\ 0 \\ \Delta t a_{\rm max} \\ \Delta t a_{\rm min}\end{bmatrix}, \\
&\mathfrak G^1 = \begin{bmatrix} 0 & 0 & 0 & 0\end{bmatrix},~
\mathfrak G^0 = \begin{bmatrix} -1 & 0 & 1 & h\end{bmatrix},~ 
\mathfrak g^0 = [0].
\end{align*}

We now describe the dynamical control system governing the follower vehicle. The state of the follower includes only the position and the velocity, $x(k) = [p_f(k),v_f(k)]^\top$, meaning that the system has a state-observation, i.e., $y(k)=x(k)$. We assume that the state evolves according to the kinematic laws:
\begin{align*}
    p_f(k+1) &= p_f(k) + \Delta t v_f(k),~\\
    v_f(k+1) &= v_f(k) + \Delta t a_f(k) + \omega(k),
\end{align*}
where $a_f(k)$ is the acceleration of the follower, and $\omega(k)$ is the process noise, which can be understood as the aggregation of exogenous forces acting upon the vehicle, e.g., wind, drag, and friction. The acceleration of the follower is taken according to the following control law:
\begin{align*}
{
    a_f(k) = \frac{p_l(k)-p_f(k) - hv_f(k)}{h\Delta t} + \frac{v_l(k) - v_f(k)}{h} - 1_{\rm m/s^2},}
\end{align*}
{in which the acceleration is dictated by the current headway, the difference in speed between the vehicles, and a constant term added to enhance robustness.}
The closed-loop system is hence governed by:
\begin{align*}
    &x(k+1) = Ax(k) + Bd(k) + w + E\omega(k),~ \omega(k)\in \mathcal{P},~\\
    &y(k) = x(k),~\mathcal{P} = \{\omega\in \R: |\omega| \le \Phi\}
\end{align*}
where:
\begin{align*}
    A = \begin{bmatrix} 1 & \Delta t \\ -\frac{1}{h} & -\frac{\Delta t}{h} \end{bmatrix},~
    B = \begin{bmatrix} 0 & 0 \\ \frac{1}{h} & \frac{\Delta t}{h} \end{bmatrix},~E = \begin{bmatrix} 0 \\ 1 \end{bmatrix},~w = \begin{bmatrix}0 \\ \Delta t\end{bmatrix}
\end{align*}
As for initial conditions, we follow Remark \ref{rem.InitDepend} and choose the set of initial conditions depending on $d(0) = [p_l(0),v_l(0)]^\top$. Namely, we assume that the headway at time $k=0$ satisfies $p_l(0) - p_f(0) - hv_f(0) \ge 0.7$.

We want to prove that the follower satisfies the contract with the given assumption and guarantees for a specific choice of parameters, and we do so by running Algorithm \ref{alg.VerifyUncertain}. We choose the parameters $\Delta t = 0.3{\rm s}$, $h = 2{\rm s}$, $a_{\rm max} = a_{\rm min} = 9.8{\rm m/s^2}$, $\Phi = 29{\rm cm}$ and a conservatism parameter $\epsilon = 10^{-12}$. 

In order to run Algorithm \ref{alg.VerifyUncertain}, we first verify that $A$ is a strictly stable matrix. The eigenvalues of $A$ can be numerically computed to be $\lambda_1 = 0$ and $\lambda_{2} = 0.85$, and all are inside the open unit disc in the complex plane. Thus the assumptions of Algorithm \ref{alg.VerifyUncertain} hold. Running the algorithm, and using Algorithm~\ref{alg.ComputeNepsi} to compute the parameter $N(\epsilon,1)$\footnote{Note that here, the matrices $\mathfrak{G}^0,\mathfrak{G}^1$ only have one row, so we need to compute only a single parameter.} and Proposition \ref{prop.TauExplicit} to compute $\tau^\epsilon$, we find that $N(\epsilon,1) = 183$ and that $\tau^\epsilon$ is given by $\tau^\epsilon = 0.58$. As instructed by Algorithm \ref{alg.VerifyUncertain}, we now run Algorithm \ref{alg.VerifyCertainIota} for the system with no perturbation, i.e., the system given by the state-space representation defined by:
\begin{align*}
&x(k+1) = Ax(k) + Bd(k) + w,\\
&y(k) = x(k),\\
&p_l(0) - p_f(0) - hv_f(0) \ge 0.6,
\end{align*}
and the robustified contract $\hat{\Omega}_{\epsilon} = (\D,\hat\Omega_\epsilon)$, where the assumptions are given by $\mathfrak{A}^1,\mathfrak{A}^0,\mathfrak{a}^0$ and the guarantees are given by $\mathfrak{G}^1,\mathfrak{G}^0,\mathfrak{g}^0-\tau^\epsilon$. The observability index $\nu$ is equal to $1$ in this case, and the depth of the LTI contract $\hat{\C}_\epsilon$ is $m=1$. Thus, $\iota = \max\{1,1\}-1 = 0$, and we are required to solve a total of $\iota+2 = 2$ optimization problems, $V_{0,0}$ and $V_{1,0}$. We use MATLAB's internal solver, {\tt linprog}, to solve the linear programs, and find that $\theta_{0,0} = -0.12 < 0$ and that $\theta_{1,0} = -0.02<0$. Thus, we conclude using Proposition \ref{prop.Taus} that the perturbed system defining the follower satisfies the contract. We also report that the algorithm was run on a Dell Latitude 7400 computer with an Intel Core i5-8365U processor, and the total runtime was $0.15$ seconds.

We demonstrate the fact that the follower satisfies the contract by simulation. We consider the following trajectory of the leader - its initial speed is about $110{\rm km/h}$, which is roughly kept for 30 seconds. It then starts to sway wildly for 30 seconds between $20-30{\rm km/h}$ and $110{\rm km/h}$, braking and accelerating as hard as possible. Finally, it stops swaying and keeps its velocity for 30 more seconds. The velocity and acceleration of the leader can be seen in Fig. \ref{fig.LeaderSimulation}(a) and \ref{fig.LeaderSimulation}(b). In particular, the leader vehicle satisfies the assumptions of the contract. The follower starts $46{\rm m}$ behind the leader, at a speed of $80{\rm km/h}$, meaning that the requirement on the initial condition is satisfied. We simulate the follower system for two cases, the first is where the noise $\omega(k)$ is adversarial, choosing the worst case value at each time, and the second is where the noise $\omega(k)$ distributes uniformly across $\mathcal{P}$. The results of the simulation can be seen in Fig. \ref{fig.LeaderSimulation}(c)-(f). In particular, it can be seen that the headway in both cases is always at least $h=2{\rm s}$, i.e., the guarantees are satisfied. 

\begin{figure*}[t]
    \centering
    \subfigure[Velocity of leader] {\scalebox{.43}{\includegraphics{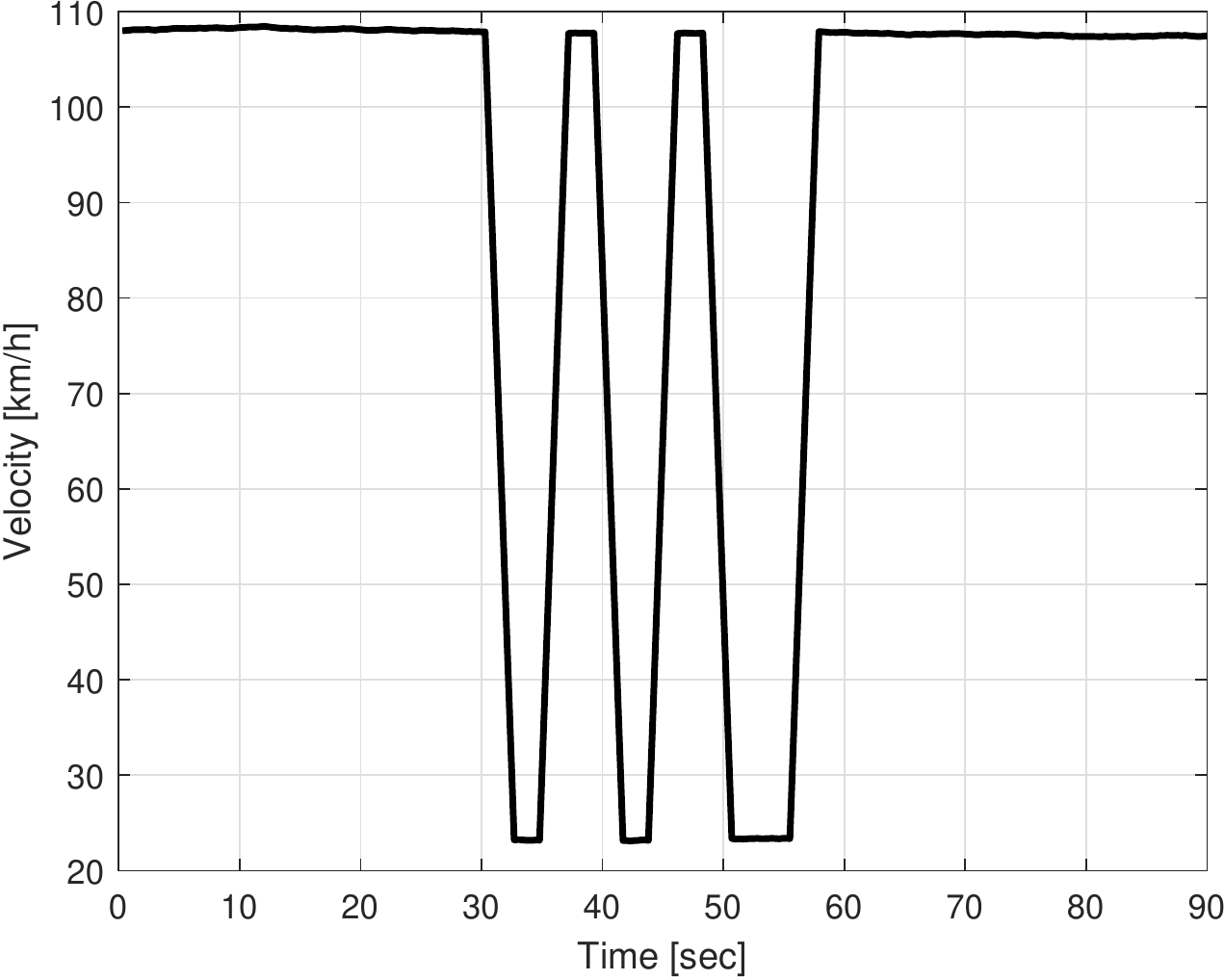}}} \hspace{.5cm}
    \subfigure[Acceleration of leader] {\scalebox{.43}{\includegraphics{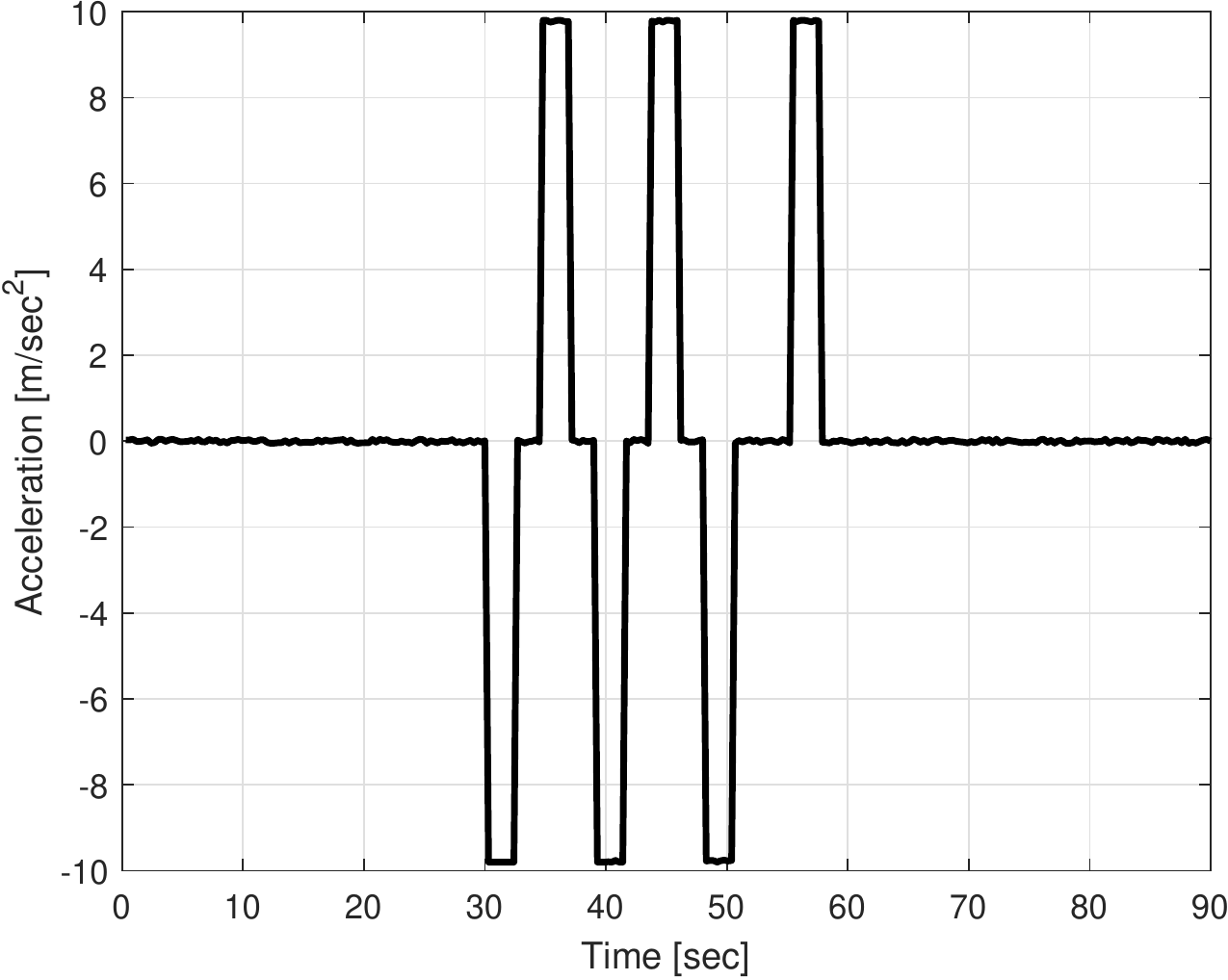}}}\hspace{.5cm}
    \subfigure[Headway] {\scalebox{.43}{\includegraphics{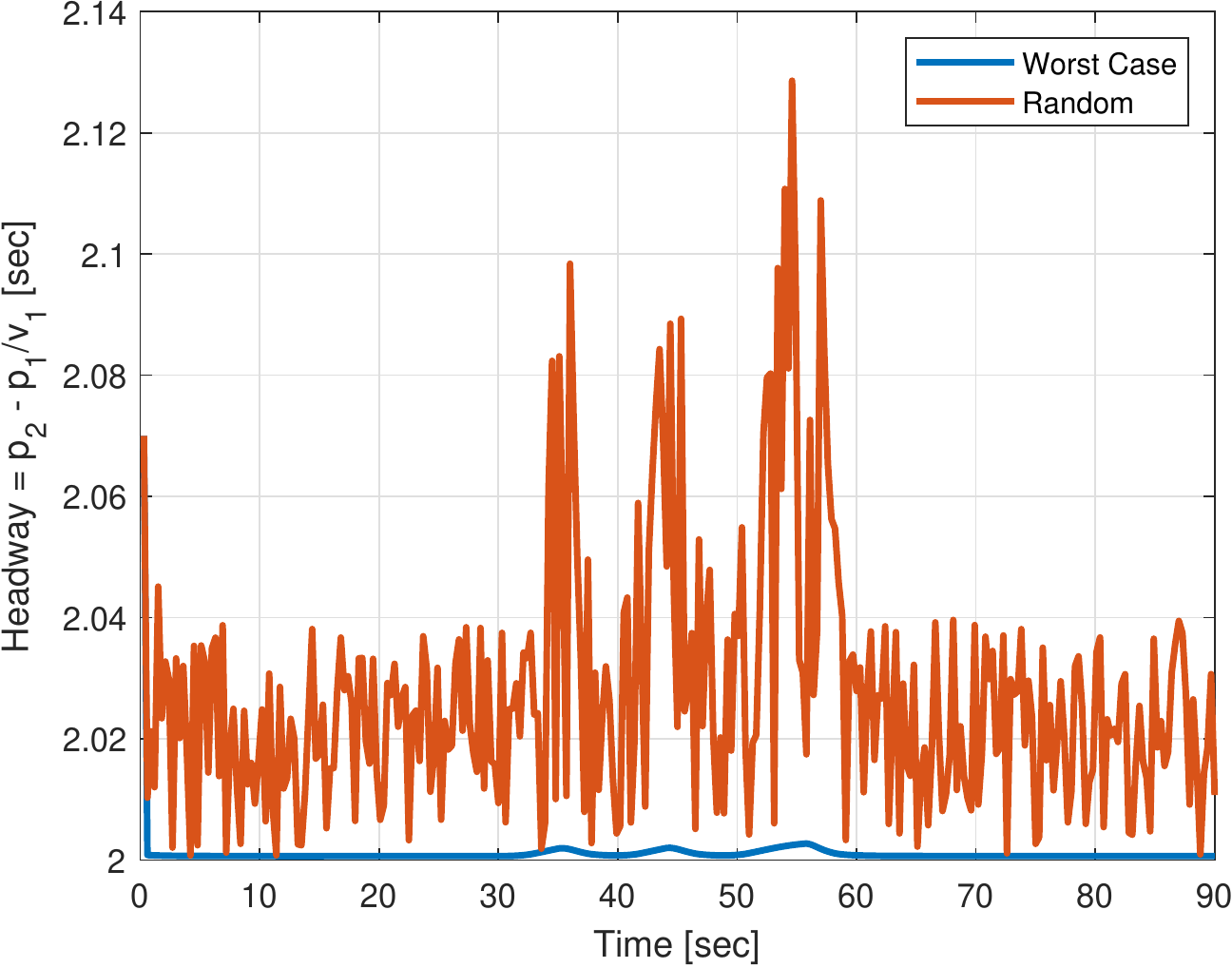}}}\hspace{2cm}
    \subfigure[Distance between the vehicles] {\scalebox{.43}{\includegraphics{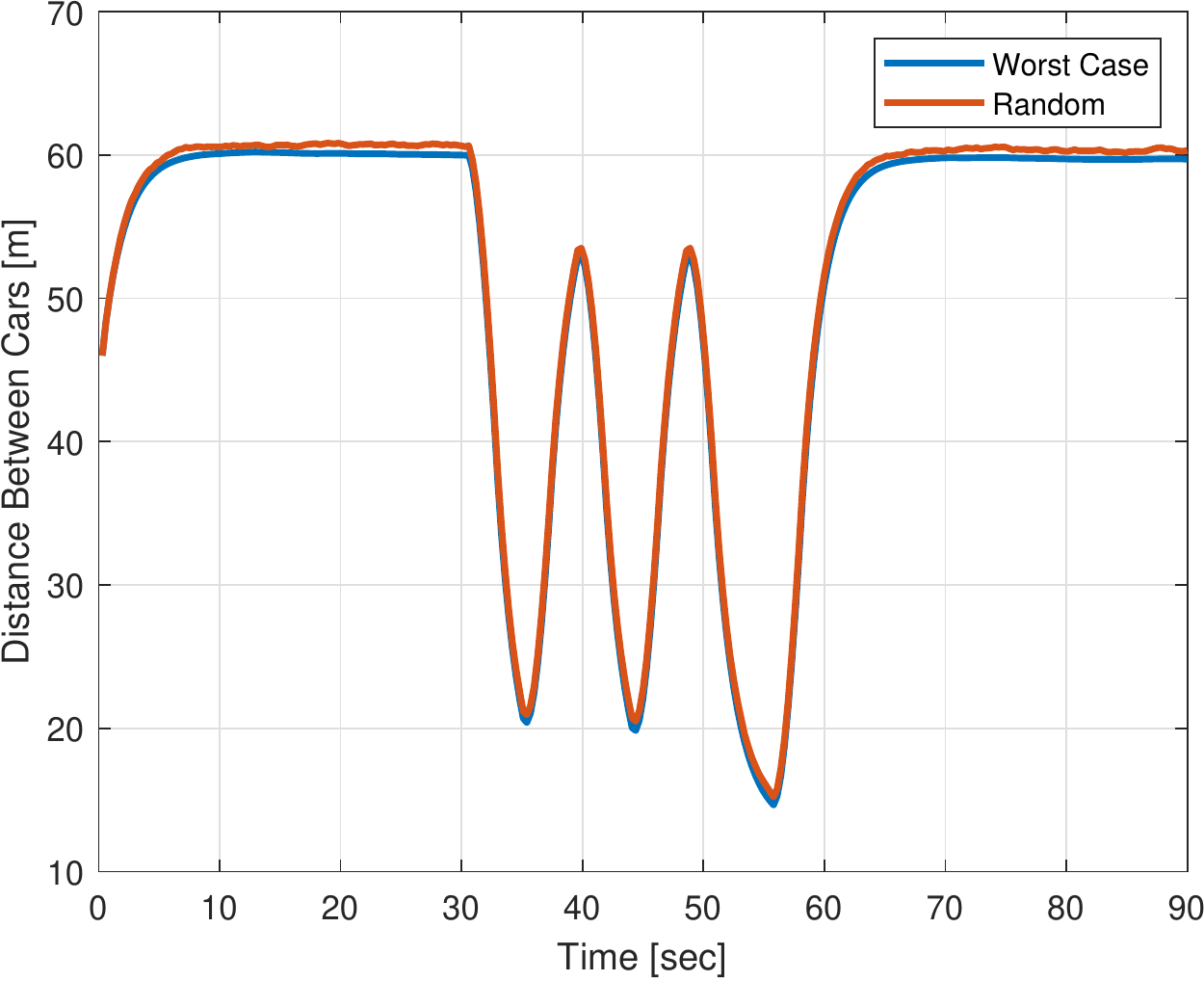}}}\hspace{.5cm}
    \subfigure[Velocity of follower] {\scalebox{.43}{\includegraphics{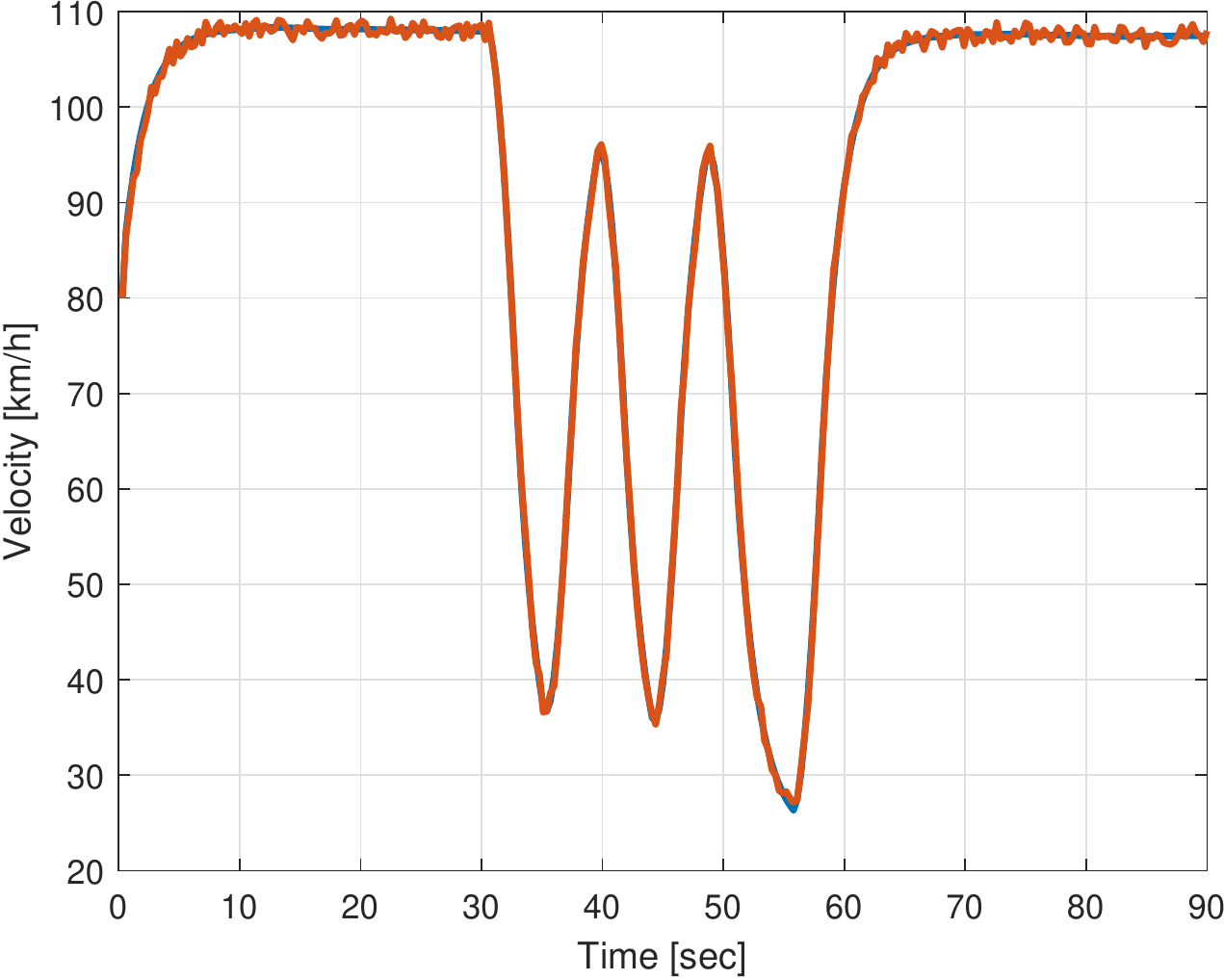}}}\hspace{.5cm}
    \subfigure[Acceleration $a_f(k)$ of follower, as dictated by the controller] {\scalebox{.43}{\includegraphics{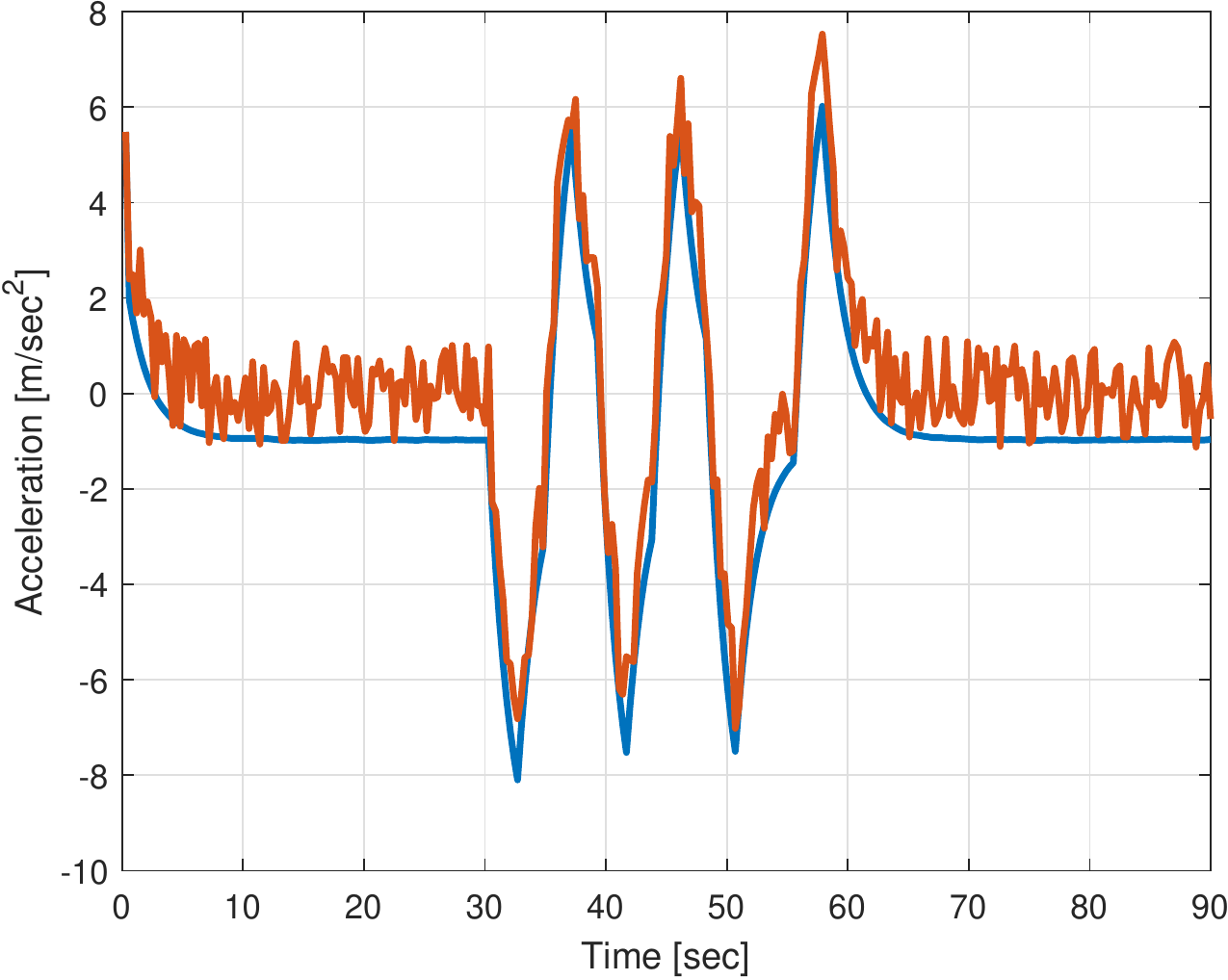}}}
    \caption{Simulation of the two-vehicle leader-follower system. The black plots correspond to the leader, the blue plots correspond to the follower with worst-case process noise, and the red plots correspond to the follower with random process noise.}
    \label{fig.LeaderSimulation}
\end{figure*}

\subsection{Formation Control for Double-Integrator Agents} \label{sec.DoubleInteg}
Formation control is a fundamental problem in the field of cooperative control, in which one tries to coordinate a collection of agents to achieve a certain spatial shape \cite{Oh2015}. This canonical problem has many versions depending on the sensing capabilities of the agents, as well as the desired degrees of freedom for the achieved shape. In all instances of the problem, the desired spatial formation is defined ``locally" by prescribing geometric constraints on each agent and agents adjacent to it, e.g., desired displacement \cite{Oh2015}, distance \cite{Oh2015}, or bearing \cite{Zhao2019}. The agents can then be maneuvered in space either by changing the geometric constraints, e.g., the desired displacement, or by assigning a few of the agents to be ``leaders", and having the other agents follow suit.

In this case study, we focus on displacement-based formation control for a directed network of double integrator agents. Our goal is to verify that a given multi-agent system satisfies a contract, in which the guarantees imply that it approximately reaches the correct spatial formation. Ideally, one would dissect this contract on the multi-agent system into smaller contracts on the individual agents. However, we run the verification process while treating the system as a monolithic entity, as our goal in this case study is to show that the methods we presented can work well even for high-dimensional systems.

We consider a network of $n_V$ $D$-dimensional agents. The system can be described using a directed graph $\mathcal{G} = (\mathcal{V},\mathcal{E})$, where the set of nodes $\mathcal{V}$ corresponds to the agents in the network, and the edges $\mathcal{E} \subset \mathcal{V} \times \mathcal{V}$ define the sensing relations between the agents. Specifically, for two nodes $i,j\in \mathcal{V}$, the edge $(i,j)$ belongs to $\mathcal{E}$ if and only if agent $i$ can measure the state of agent $j$. We let $n_E = |\mathcal{E}|$ be the number of edges in the graph. 

The state of the $i$-th agent is given by $[p_i,v_i]$, where $p_i \in \R^D$ is the position of the agent, and $v_i \in \R^D$ is its velocity. We choose one agent, denoted as $1 \in \mathcal{V}$, to be the leader node, so it will move independently from all other agents, which will follow it in space while trying to keep the desired spatial shape. The input to the system is then given by $d = [a_1,\delta]$ where $a_1 \in \R^D$ is the acceleration of the leader node, and $\delta \in \R^{n_ED}$ is a stacked vector consisting of the desired displacements. More precisely, for each edge $(i,j) \in \mathcal{E}$, the vector $\delta_{ji} \in \R^D$ is the desired relative displacement from the $j$-th agent to the $i$-th agent. The output from the system consists of the positions, relative to the leader, i.e., $y = (p_i - p_1)_{i\in \mathcal{V}, i\neq 1}$\footnote{We choose the output as the relative position to avoid strict stability issues later, as formation control protocols are invariant to translating all of the agents in the same direction and by the same amount.}. The guarantees we want to make are that the agents' displacements are close to the desired ones. Namely, we wish to guarantee that $-(\mu_{\rm err})_{ij} \le p_i(k) - p_j(k) - \delta_{ji}(k) \le (\mu_{\rm err})_{ij}$ holds at any time $k\in \N$, where $\mu_{\rm err} \in \R^{Dn_E}$ is a constant vector defining the allowable error for each pair $(i,j)\in \mathcal E$. The entries $(\mu_{\rm err})_{ij}$ of the vector $\mu_{\rm err}$ can be chosen arbitrarily. However, if the graph $\mathcal{G}$ is a directed acyclic graph with large diameter, it is advisable to take the entries of $\mu_{\rm err}$ as different from one another, due to string-stability-like phenomena \cite{Feng2019}. 

As for the assumptions, a reasonable assumption on $a_1 \in \R^D$ can bound the maximum acceleration and deceleration of the agent in each spatial direction, i.e., $a_i(k) \in [-a_{\rm min},a_{\rm max}]^D$. As for the desired displacements $(\delta_{ij})_{(i,j)\in\mathcal{E}}$, we make two assumptions. First, we assume that the desired displacements can only change by a bounded amount between time iterations. Namely, we assume that $\|\delta_{ij}(k+1) - \delta_{ij}(k)\|_\infty \le \mu_{\rm diff}$ for any $(i,j)\in \mathcal{E}$ and any time $k\in \N$, where $\|\cdot\|_\infty$ is the sup-norm. Moreover, we assume that the desired displacements $(\delta_{ij})_{(i,j)\in\mathcal{E}}$ are \emph{consistent} with one another, i.e., that there exists a configuration in space attaining these displacements. If we let $E \in \R^{n_V\times n_E}$ be the incidence matrix of the graph $\mathcal{G}$, this demand is equivalent to $\delta(k) \in {\rm Im}(E^\top \otimes {\rm I}_D)$, were ${\rm I}_D \in \R^{D\times D}$ is the identity matrix and $\otimes$ is the Kronecker product. By using the SVD decomposition\footnote{More precisely, if $E^\top = U\Sigma V^\top$ is the SVD decomposition, we define $P = \tilde{\Sigma}U^\top$, where $\tilde{\Sigma} \in \R^{n_E \times n_E}$ is a diagonal matrix satisfying $\tilde{\Sigma}_{ii} = 1$ if and only if $\Sigma_{ii} = 0$, and $\tilde{\Sigma}_{ii} = 0$ otherwise.} of $E^\top$, we build a matrix $P \in \R^{n_E\times n_E}$ such that $\ker(P) = {\rm Im}(E^\top)$, and we can restate the consistency assumption as $\left[\begin{smallmatrix} P \otimes {\rm I}_D \\ -P\otimes{\rm I}_D \end{smallmatrix}\right]^\top \delta(k) \le 0$ for any time $k\in \N$. In particular, the contract defined by the assumptions and guarantees is LTI of depth $1$.

As for the system, we assume that the agents are double integrators, where all non-leader agents follow a linear control law. Namely, we assume that the position and velocity of the $i$-th agent evolve according to the following equations
\begin{align*}
    p_i(k+1) &= p_i(k) + \Delta t v_i(k),~\\v_i(k+1) &= v_i(k) + \Delta t a_i(k) + \omega_i(k),
\end{align*}
where the noise $\omega_i(k) \in \R^D$ corresponds to unmodeled forces on the agent, and we assume that the Euclidean norm of $\omega_i(k)$ is bounded by a tunable parameter $\omega_{\rm max}$.
Moreover, the control input $a_i(k)$ for $i\neq 1$ is given by the following linear law
\begin{align*}
    a_i(k) = \frac{1}{d_{i}^{\rm out}} \sum_{j: (i,j)\in \mathcal{E}}\left(-\frac{p_i-p_j-\delta_{ji}}{\Delta t^2} - 2\frac{v_i-v_j}{\Delta t}\right),
\end{align*}
where $d_i^{\rm out}$ is the out-degree of the node $i$, i.e., the number of agents $j$ such that $(i,j)\in \mathcal{E}$. Unfortunately, the equations above define an LTI system which is not strictly stable, as the system matrix $A$ has $2D$ eigenvectors with eigenvalue $\lambda = 1$, namely $\left[\begin{smallmatrix}{\rm e}_i \otimes {\rm I}_D \\ 0\end{smallmatrix}\right]$ and $\left[\begin{smallmatrix} 0\\ {\rm e}_i \otimes {\rm I}_D \end{smallmatrix}\right]$. These correspond to moving all agents in the same direction and by the same amount, and to adding the same vector to all of the agents' velocities, correspondingly. To overcome this problem and make Algorithm \ref{alg.VerifyUncertain} applicable for this problem, we define $2(n_V-1)$ new coordinates as $q_i = p_i - p_1$ and $u_i = v_i - v_1$ for $1\neq i\in \mathcal{V}$. A simple calculation shows that $q,u$ evolve according to the following equations:
\begin{align*}
    q_i(k+1) &= q_i(k) + \Delta t u_i(k),\\
    u_i(k+1) &= u_i(k) + \Delta t a_i(k) + \omega_i(k) - \Delta t a_1(k),
\end{align*}
where the control input is given by
\begin{align*}
    a_i(k) &= \frac{1}{d_{i}^{\rm out}} \sum_{j: (i,j)\in \mathcal{E}}\left(-\frac{q_i-q_j-\delta_{ji}}{\Delta t^2} - 2\frac{u_i-u_j}{\Delta t}\right),
\end{align*}
where we define $q_1 = u_1 = 0 \in \R^D$, and the output of the system is, as before, given by $y = q$. Thus, this is a perturbed LTI system with observability index equal to $\nu = 2$.

\begin{table*}[!ht]
\begin{center}
\begin{tabular}{ |c|c|c|c|c|c|>{\centering\arraybackslash}m{1.5cm}|>{\centering\arraybackslash}m{1.5cm}|c|c|c| } 
\hline
$n_V$ & $n_E$ & Graph Type & System dim. & Input dim. & Output dim. & Number of Assumptions & Number of Guarantees & Alg. \ref{alg.ComputeNepsi} Time & LP Time & Total Time\\ \hline
5 & 10 & Complete & 16 & 22 & 8 & 84 & 40 & 0.03 & 0.48 & 0.51\\ \hline
10 & 45 & Complete & 36 & 92 & 18 & 364 & 180 & 0.41 & 2.88 & 3.29 \\ \hline
15 & 105 & Complete & 56 & 212 & 28 & 844 & 420 & 1.24 & 13.24 & 14.49 \\ \hline
20 & 190 & Complete & 76 & 382 & 38 & 1524 & 760 & 7.73 & 60.22 & 67.96 \\ \hline
30 & 435 & Complete & 116 & 872 & 58 & 3484 & 1740 & 76.48 & 527.61 & 604.09 \\ \hline
50 & 1225 & Complete & 196 & 2452 & 98 & 9804 & 2900 & 1532.81 & 9740.66 & 11273.47 \\\hline
30 & 30 & Cycle & 116 & 62 & 58 & 244 & 120 & 2.46 & 2.93 & 5.39 \\ \hline
50 & 50 & Cycle & 196 & 102 & 98 & 404 & 200 & 30.29 & 8.48 & 38.78 \\ \hline 
\end{tabular}
\end{center}
\caption{An analysis of the runtime (in seconds) of Algorithm \ref{alg.VerifyUncertain} with Algorithm \ref{alg.ComputeNepsi} for the formation control problem, for $D=2$. Here, LP Time refers to the time (in seconds) it took to compute all parameters $\theta_{n,\ell}$ needed by the Algorithm \ref{alg.VerifyUncertain}.}
\label{table.Runtime}
\end{table*}

In order to verify whether the system satisfies the given contract, we choose certain values for the tunable parameters $\Delta t, a_{\rm max}, a_{\rm min}, \mu_{\rm diff},\mu_{\rm err}$ and $\omega_{\rm max}$, and run Algorithm \ref{alg.VerifyUncertain} with Algorithm \ref{alg.ComputeNepsi} and $\epsilon = 10^{-12}$. The algorithms were executed on a Dell Latitude 7400 computer with an Intel Core i5-8365U processor for multiple values of $n_V$ and different graphs $\mathcal{G}$. The runtimes are reported in Table \ref{table.Runtime}. The table concerns two distinct cases. In the first, the graph $\mathcal{G}$ is chosen as a complete graph on $n_V$ nodes. In this case, the runtime of the algorithm is about 10 minutes even for systems of order exceeding to 100, with thousands of assumptions and guarantees. Moreover, we can check whether the system satisfies the contract in about three hours even for systems of order roughly equal to $200$, with almost $10000$ assumptions and a few thousand guarantees.

In the second case, we choose the graph $\mathcal{G}$ by taking agents $\mathcal V = \{1,2,\ldots,n_V\}$, and taking a total of $n_V$ edges defined as follows - we take $(i+1,i)$ for all $i=1,2,\ldots,\lfloor n_V/2 \rfloor$, we also take $(i,i+1)$ for $i=\lfloor n_V/2 \rfloor+1,\ldots,n_V-1$, and lastly, we also take $(n_V,1)$. One can see the graph $\mathcal{G}$ as a union of two paths, of lengths $\lfloor n_V/2 \rfloor$ and $\lceil n_V/2 \rceil$, which coincide only at the first and the last node. The graph $\mathcal{G}$ can also be seen as a cycle, where we change the orientation of some of the edges. In this case, the matrices defining the system are sparse. As expected, the algorithm runs significantly faster in this case, terminating in under a minute even for a system of order roughly equal to $200$.

\subsection{Discussion}
We considered two numerical examples. The first numerical example considered a low-dimensional LTI system with interval uncertainty, whereas the second considered a very high-dimensional system with non-polyhedral constraints on the perturbation. The runtimes reported in Table \ref{table.Runtime} demonstrate the applicability of our approach even for extremely large systems and for specifications with many assumptions and guarantees. 


We also compare our approach with other formal verification techniques.
Trying to apply classical model-checking tools would first require us to build an abstraction of the system, which is a finite transition system \cite{Belta2017}. This abstraction is almost always achieved either by discretizing the state space, by defining an equivalence relation using the signs of the values of the functions defining the guarantees, or by further refining either of the two. For the numerical example in Section \ref{sec.DoubleInteg} with $n=50$ vertices and a cycle graph, both approaches result in finite transition systems with roughly $10^{60}$ discrete states, rendering this approach as highly inapplicable.

Other approaches for verification rely on approximate simulation and bi-simulation, see \cite{Girard2007}. These methods first quantify the distance between the system-under-test and a lower-dimensional system, and then solve the verification problem for the latter using other methods, e.g., discretization-based model checking or reachability analysis. However, the standard definition of bi-simulation cannot incorporate assumptions on the input other than $u(k)\in \mathcal{U}, \forall k\in \N$, and thus cannot be used for verifying specifications defined by LTI contracts of depth $m\ge 1$. Once bi-simulation will be properly extended to incorporate non-static assumptions on the input, it could be coupled with the theory presented in this work.

\section{Conclusions}
In this paper, we presented a framework for verifying assume/guarantee contracts defined by time-invariant linear inequalities for perturbed LTI systems. First, we defined the notion of LTI contracts of an arbitrary depth $m$. Second, we generalized the results of \cite{SharfADHS2020} and provided an LP-based mechanism for verifying that a given unperturbed LTI system satisfies a general LTI contract of arbitrary depth $m$, namely Algorithm \ref{alg.VerifyCertainIota}. Third, we presented a comparison-based mechanism for verifying that a perturbed LTI system $\Sigma$ satisfies an LTI contract of arbitrary depth. Namely, we showed that a perturbed system satisfies a contract with linear-time invariant guarantees if and only if the nominal version of the system (with no perturbations) satisfies a robustified version of the contract. Unfortunately, this robustified contract is time-varying, so we refined it by a tractable LTI contract, and then applied the LP-based tools for unperturbed systems to check whether the nominal LTI system satisfies it. This discussion resulted in Algorithm \ref{alg.VerifyUncertain}, and the correctness, the assumptions, the computational complexity and the approximation properties of the algorithm were studied. We exhibited the tools developed in two case studies, one considering autonomous driving, and one considering multi-agent systems.
Future research can try and derive LP-based verification methods for a wider class of systems, including LTI hybrid systems, perturbed hybrid systems, and uncertain systems.
Another possible avenue for future research is building semi-definite programming-based tools for contracts defined using quadratic or LMI-based inequalities. Lastly, one could try to construct LP-based tools supporting the modular framework of contract theory, namely refinement and composition, extending the tools presented in \cite{SharfADHS2020}.

\bibliographystyle{ieeetr}
\bibliography{main}

\appendix
This appendix is dedicated to the proof of the second part of Theorem \ref{thm.AssumptionTau}. We start by stating and proving a few lemmas:
\begin{lem} \label{lem.2}
If $w\in \R^n$ is not contained in the stable subspace of $A$, then there exists a constant $c>0$ such that for any $\varsigma\in \N$, we have
$\max \left\{w A^\varsigma \kappa: ~\|\kappa\|\le 1\right\} \ge c.$
\end{lem}
\begin{proof}
By assumption, there exists an eigenvector $v \in \CC^n$ of $A$ with eigenvalue $\lambda \in \CC$ such that $v^\top w \neq 0$, $\|v\| = 1$ and $|\lambda|\ge 1$. As $A$ is a matrix with real entries, we have that $A\bar{v} = \bar{\lambda}\bar{v}$, where $\bar{\cdot}$ denotes the complex conjugate.
We let $v_R$ be the real part of $v$, and $v_I$ be the imaginary part of $v$. As $A^{\varsigma} v = \lambda^\varsigma v$ and $A^\varsigma\bar{v} = \bar{\lambda}^\varsigma\bar{v}$, we conclude that:
\begin{align*}
    w^\top A^\varsigma v_R = {\rm Re}(\lambda^\varsigma w^\top v),~~
    w^\top A^\varsigma v_I = {\rm Im}(\lambda^\varsigma w^\top v) 
\end{align*}
Moreover, it is clear that $\|v_R\|^2 + \|v_I\|^2 = \|v\|^2 = 1$, and in particular that $\|v_R\|,\|v_I\| \le 1$. Thus, by choosing $\kappa = \pm v_R, \pm v_I$ we conclude that:
\begin{align*}
    \max \left\{w A^\varsigma \kappa:\|\kappa\|\le 1\right\} \ge \max\{|{\rm Re}(\lambda^\varsigma w^\top v)|,|{\rm Im}(\lambda^\varsigma w^\top v)|\}
\end{align*}
Now, as $|{\rm Re}(\lambda^\varsigma w^\top v)|^2 + |{\rm Im}(\lambda^\varsigma w^\top v)|^2 = |\lambda^\varsigma w^\top v|^2$, we conclude that the right-hand side is at least as big as $\frac{|\lambda^\varsigma w^\top v|}{\sqrt{2}}$. We choose $c = \frac{|w^\top v|}{\sqrt{2}}$ and conclude the proof as $|\lambda| \ge 1$.
\end{proof}

\begin{lem} \label{lem.3}
Let $v_1,\ldots,v_N \in \R^n$ be vectors and $b_1,\ldots,b_N\in \R$ be scalars. Define the set $Q = \{x\in \R^n : v_i^\top x \le b_i, \forall i\}$, which is assumed to be non-empty. The set $Q$ is compact if and only if for any unit vector $\xi$ there exists some $i$ such that $v_i^\top \xi > 0$.
\end{lem}
\begin{proof}
Suppose first that $Q$ is compact, and fix some $x_0 \in Q$. Taking an arbitrary unit vector $\xi$, the set $Q$ cannot contain the ray $\{x_0 + t\xi\}_{t>0}$, as it is non-compact. Thus, for some $t>0$ and some $i$, we must have $v_i^\top(x_0 + t\xi) > b_i$. As $v_i^\top x_0 \le b_i$, we conclude that $v_i^\top \xi > 0$.

On the contrary, suppose now that for any unit vector $\xi$ there exists some $i$ such that $v_i^\top \xi > 0$. The set $Q$ is closed by definition, so it suffices to show that it is bounded. If this is not the case, then there exists a sequence $\{x_j\}_{j=1}^\infty \in Q$ with $\|x_j\| \to \infty$. Taking some $x_0\in Q$, we use the compactness of the unit ball in $\R^n$ find a subsequence $\{x_{n_k}\}_{k=1}^\infty$ such that the sequence of unit vectors $\{\frac{x_{n_k}-x_0}{\|x_{n_k}-x_0\|}\}_{k=1}^\infty$ converges to some unit vector $\xi$. It is easy to see that because $Q$ is convex and closed, it must contain the ray $\{x_0 + t\xi\}_{t>0}$, However, as we saw above, this is not possible as there exists some $i$ such that $\xi^\top v_i > 0$. We arrived at a contradiction, and therefore conclude that $Q$ must be compact.
\end{proof}

\begin{lem} \label{lem.4}
Let $v_1,\ldots,v_N \in \R^n$ be vectors and $b_1,\ldots,b_N\in \R$ be scalars. Define the set $Q = \{x\in \R^n : v_i^\top x \le b_i, \forall i\}$. If the set $Q$ is compact, then there exists some $M\ge 0$ such that for any $c_1,c_2,\ldots,c_N \ge 0$, if $\max_i c_i > M$ then
the set $Q^\prime = \{x\in \R^n : v_i^\top x \le b_i - c_i, \forall i\}$ is empty.
\end{lem}
\begin{proof}
If $Q$ is empty, we take $M = 0$. Otherwise, for any $i$, we define $M_i = b_i - \min_{x\in Q}v_i^\top x + 1$. The minimum is finite as the set $Q$ is compact. Moreover, it is clear by definition that for any $x\in Q$ and for any $i$, we have $v_i^\top x > b_i - M_i$. 

Take $M = \max_i M_i$. If $c_1,c_2,\ldots,c_N \ge 0$ and $\max_i c_i > M$, then the set $\{x\in \R^n : v_i^\top x \le b_i - c_i, \forall i\}$ is a subset of $Q$. However, there exists some $i_0$ such that $c_{i_0} > M_{i_0}$, so for any $x\in Q$, we have $v_{i_0}^\top x > b_{i_0}- M_{i_0}$. Thus the set $\{x\in \R^n : v_i^\top x \le b_i - c_i, \forall i\}$ cannot contain any points from $Q$, hence it is empty.
\end{proof}

We now prove the second part of Theorem \ref{thm.AssumptionTau}
\begin{proof}
We assumed that $\PP$ contains a neighborhood of the origin. As the matrix $E$ has full row rank, we conclude that the image of $\PP$ under $E$ also contains a neighborhood of the origin, denoted as $\{x: \|x\| \le \delta\}$ for some $\delta > 0$. By assumption, the image of $T^\top$ is not contained within the stable subspace of $A$. Thus, there exists some $i$ such $T^\top {\rm e}_i$ is not inside the stable subspace of $A$. Thus, by Lemma \ref{lem.2}, we conclude that there exists some constant $c>0$ such that:
\begin{align*}
    \tau^{\mathcal P,{\rm m},\varsigma}_i &= \max \left\{{\rm e}_i^\top T A^\varsigma E \omega: ~\omega \in \mathcal{P}\right\} \\&\ge \max \left\{{\rm e}_i^\top T A^\varsigma \kappa: ~\|\kappa\|\le \delta\right\} \ge c\delta.
\end{align*}
In particular, the $i$-th entry of $\tau^k$ grows unbounded as $k\to \infty$. Now, consider the set $Q$ defined as:
\begin{align*}
    Q = \left\{(y_0,\ldots,y_m) : \sum_{r=0}^m \mathfrak G^r \begin{bmatrix} d_r \\ y_r \end{bmatrix} \le \mathfrak g^0\right\}.
\end{align*}
We assumed that the set is bounded and non-empty for some fixed $d_0 = \check d_0,\ldots,d_m = \check d_m$. The set $Q$ can be equivalently written as:
\begin{align*}
    Q = \left\{(y_0,\ldots,y_m) : \sum_{r=0}^m \mathfrak G^r_y y_r \le \mathfrak g^0 - \sum_{r=0}^m \mathfrak G^r_d d_r\right\}.
\end{align*}
where $\mathfrak G^r = [\mathfrak G^r_d,\mathfrak G^r_y]$. Thus, by Lemma \ref{lem.3}, this set is compact for \emph{any} choice of $d_0,\ldots,d_r$, as the condition for compactness depends only on the left-hand side of the linear inequality defining $Q$.

Now, fix some $d_0,d_1,\ldots,d_m$ which are compatible with the inequality defining the set of assumptions $\D$. By Lemma \ref{lem.4}, we conclude that there exists some $M>0$ such that if $c_i$ satisfy $\max_i{c_i} > M$, then the following set is empty, where $c = (c_i)$:
\begin{align*}
    Q^\prime = \left\{(y_0,\ldots,y_m) : \sum_{r=0}^m \mathfrak G^r_y y_r \le \mathfrak g^0 - \sum_{r=0}^m \mathfrak G^r_d d_r - c\right\},
\end{align*}
Taking $c = \tau^k$ for a large enough $k$, we know that $\max_i{c_i} > M$, meaning that the set $Q^{\prime}$ is empty. However, $Q^{\prime}$ has an equivalent formulation:
\begin{align*}
    Q^\prime = \left\{(y_0,\ldots,y_m) : \sum_{r=0}^m \mathfrak G^r \begin{bmatrix} d_r \\ y_r \end{bmatrix} \le \mathfrak g^0 - \tau^k\right\}.
\end{align*}
Therefore, we conclude that no choice of $y_0,\ldots,y_r$ can satisfy the guarantees of $\C^\prime$ for the input $d_0,\ldots,d_r$. In particular, the system $\hat{\Sigma}$ cannot satisfy $\C^\prime$, and thus $\Sigma \not\sat \C$.
\end{proof}

\end{document}